\documentclass[11pt]{article}
\usepackage{amsmath,amsfonts,amssymb,amsthm,latexsym, epsfig, graphics,mathrsfs, mathpazo,dsfont,wrapfig,bbm, fixmath, bm,multirow}
\usepackage{booktabs} \usepackage{longtable} \usepackage{array} \usepackage{amsmath,amssymb}
\usepackage{courier}
\usepackage[margin=1in,vmargin=1in]{geometry}
\usepackage{color}

\usepackage[numbers]{natbib}
\newtheorem{theorem}{Theorem}
\newtheorem{lemma}{Lemma}
\newtheorem{proposition}{Proposition}

\newtheorem{corollary}{Corollary}

   \theoremstyle{plain}
    \numberwithin{equation}{section}

\newtheorem{remark}[theorem]{Remark}

\usepackage{xcolor}

\definecolor{Blue}{rgb}{0,0,1}

\definecolor{Green}{rgb}{0,1,0}

\definecolor{Magenta}{rgb}{1,0,1}

\begin{document}

\title{A Spread-Gated Hawkes-Flocking Model for Best Bid and Ask Dynamics, with an Application to Limit Order Placement}
\author{Hyoeun Lee\footnote{Department of Statistics, University of Illinois, Email: hyoeun@illinois.edu}
,\,\,Kiseop Lee\footnote{Department of Statistics, Purdue University, Email: kiseop@purdue.edu}
}
\maketitle

\begin{abstract}
	We study the joint dynamics of the best bid and ask prices with a spread-gated Hawkes-flocking model. The model tracks four types of best-quote movements: spread-narrowing movements are switched off when the spread is at its one-tick minimum, and a cross-side excitation term, whose activation depends on the prevailing spread, links the two sides of the book. We show that the process is non-explosive on every finite horizon, give an $O(N)$ recursive likelihood, and validate the maximum likelihood estimator by simulation. On real intraday limit order book data for two large-tick stocks, INTC and MSFT, the restriction that removes the cross-side term is rejected, and the full model improves fit substantially by AIC and BIC; the likelihood is multimodal on a single day, so estimation uses a multi-start search. As an application, we derive the closed-form optimal size of a single-period limit order placed at the best or second-best quote, given the model's next-event probabilities and externally supplied execution probabilities.
\end{abstract}

\section{Introduction}

Hawkes processes are a standard tool for modeling the clustering of events in high-frequency financial data; \citet{bacry2015hawkes} review their use in finance. In limit order book (LOB) modeling they have been used to describe market order flow and its market impact~\cite{alfonsi2016dynamic}, the arrivals of market orders, limit orders and cancellations in a full order book~\cite{abergel2015long,horst2019scaling}, and price variations~\cite{bacry2013modelling}.

A smaller literature models the best quotes or the bid--ask spread directly. \citet{zheng2014modelling} describe the movements of the best bid and best ask with a multivariate Hawkes process with constraints that keep the spread at least one tick wide. \citet{lee2023modeling} extend this line with a spread-dependent baseline and a stochastic reset of the excitation that keeps intensities nonnegative, and estimate the model daily on consolidated quotes of high-priced stocks. \citet{ruan2023selfexciting} model spread jumps of several sizes, multiplying the intensity by a function of the current spread. Other Hawkes-type models let the dynamics depend on the state of the book more generally. In the state-dependent Hawkes process of \citet{morariupatrichi2022state}, applied to Nasdaq order-flow data, the kernel depends on the state (such as the spread or the queue imbalance) at the time of the exciting event. \citet{sfendourakis2020lob} multiply a Hawkes component by a factor of the current state and, for large-tick stocks, use an indicator of a one-tick spread. \citet{wu2020single} and \citet{protter2024order} let baseline intensities depend on queue sizes and liquidity states, \citet{jain2024limit} combine compound Hawkes processes for order sizes with spread-dependent in-spread intensities, and \citet{kirchner2022hawkes} model market orders, limit orders and cancellations with a marked, self- and cross-exciting Hawkes process whose baseline intensity depends on the order-book imbalance. Selecting the nonzero excitation terms and kernel shapes nonparametrically, they find that imbalance predicts the side of the next market order in an essentially linear way. That the intensity of market orders rises when the spread falls to one tick is an empirical regularity of large-tick stocks~\cite{munitoke2017modelling,sfendourakis2020lob}; \citet{jain2025tick} study the role of tick size in order book dynamics more generally.

This paper builds on the constrained Hawkes framework of \citet{zheng2014modelling} and the flocking specification of \citet{jang2020systemic}. We model four types of best-quote movements (the best ask moving up or down, and the best bid moving up or down) with a multivariate Hawkes intensity that has three ingredients: a self- and mutually-exciting kernel $\Phi$ within each side of the book, a cross-side kernel $\Psi$ whose entries are activated according to the current spread, and an outer gate that sets the intensities of spread-narrowing movements to zero when the spread is at its one-tick minimum. Our model differs from these in two ways. First, the hard gate that removes narrowing movements at one tick is combined with a cross-side kernel whose activation, row by row, depends on the current spread. Second, all kernels are nonnegative and the state enters only through a $0$--$1$ gate, so the intensity is dominated by that of a linear Hawkes process. We use this to show non-explosion for the constrained process itself (Section~2.3).

We make three contributions. First, a model and a non-explosion result for it (Section~2). Second, an $O(N)$ recursive likelihood, validated by simulation, and an application to LOBSTER data for two large-tick stocks (Section~3.3), where the spread sits at one tick more than $99\%$ of the time. We test formally whether the cross-side term is needed. The likelihood is multimodal on a single day, so a multi-start search is needed; the observed information at the best fit is well conditioned once parameters are scaled (Appendix~B). We focus on large-tick stocks because there the spread is essentially a two-state process, so a binary gate is the relevant state dependence. Third, a deliberately modest trading application: a single-period, mean--variance limit-order placement problem, for which the model supplies next-event probabilities and which has a closed-form optimal order size (Section~4).

Hawkes processes have also been used in trading problems. \citet{cartea2014modelling} link mutually exciting market-order flow to short-term alpha and adverse selection, \citet{jusselin2021optimal} studies market making under persistent order flow, and \citet{choi2021optimal} consider market making when buy and sell arrivals are synchronized. For limit-order placement, \citet{guo2013optimal} derive optimal static and dynamic placement in a correlated random walk model of the best quotes, \citet{guilbaud2013optimal} treat optimal high-frequency trading with limit and market orders, \citet{lehalle2017limit} analyze placement with adverse selection and latency, and \citet{figueroalopez2018optimal} model queue positions. Our application is simpler than these: it isolates one placement decision, and treats the execution probabilities as inputs.

The paper is organized as follows. In Section~2, we define the model, relate it to existing models, and show non-explosion. In Section~3, we describe estimation, the simulation validation, and the application to real data. In Section~4, we develop the single-period placement problem. We conclude in Section~5, and the appendices contain additional proofs and real-data estimation results.

\newpage

\subsection{Notation Summary} 

Tables~\ref{tab:state_variables}--\ref{tab:assumptions} summarize the notation used throughout the paper.
\begin{table}[ht] 
\centering 
\caption{State Variables} 
\begin{tabular}{ll} 
\toprule Symbol & Description \\ 
\midrule $A(t)$ & Best ask price at time $t$ \\ 
$B(t)$ & Best bid price at time $t$ \\ 
$S(t)=A(t)-B(t)$ & Bid--ask spread \\ 
$\delta$ & Tick size (minimum spread level) \\ 
$X(t)$ & Cash position \\ $Y(t)$ & Inventory position \\ 
$G(t)$ & Mark-to-market wealth \\
& $G(t)=X(t)+\frac{A(t)+B(t)}{2}Y(t)$ \\ 
\bottomrule 
\end{tabular} \label{tab:state_variables} 
\end{table} 

\begin{table}[ht] 
\centering 
\caption{Counting Processes for Price Movements} 
\begin{tabular}{ll} 
\toprule Symbol & Description \\ 
\midrule $N_A^u(t)$ & Number of upward movements of the best ask price \\ $N_A^d(t)$ & Number of downward movements of the best ask price \\ $N_B^u(t)$ & Number of upward movements of the best bid price \\ 
$N_B^d(t)$ & Number of downward movements of the best bid price \\ $
N(t)$ & $\bigl( N_A^u(t), N_A^d(t), N_B^u(t), N_B^d(t) \bigr)^\top$ \\ \bottomrule 
\end{tabular} \label{tab:counting_processes} 
\end{table} 

\begin{table}[ht] 
\centering \caption{Conditional Intensities} 
\begin{tabular}{ll} \toprule Symbol & Description \\ 
\midrule $\lambda_A^u(t)$ & Intensity of an upward ask-price movement \\ $\lambda_A^d(t)$ & Intensity of a downward ask-price movement \\ $\lambda_B^u(t)$ & Intensity of an upward bid-price movement \\ $\lambda_B^d(t)$ & Intensity of a downward bid-price movement \\ $\lambda(t)$ & $\bigl( \lambda_A^u(t), \lambda_A^d(t), \lambda_B^u(t), \lambda_B^d(t) \bigr)^\top$ \\ 
\bottomrule 
\end{tabular} \label{tab:intensities} 
\end{table} 

\begin{table}[ht] 
\centering \caption{Parameters of the Hawkes-Flocking Model} \begin{tabular}{ll} \toprule Symbol & Description \\ \midrule $\mu_i$ & Baseline intensity of component $i$ \\ $\alpha_{ij}$ & Excitation magnitude parameter \\ $\beta_{i}$ & Exponential decay parameter \\ $\Phi(t)$ & Kernel matrix \\ $\Psi$ & Flocking interaction matrix \\ $k$ & Base excitation matrix \\  $\rho(\cdot)$ & Spectral radius \\ \bottomrule \end{tabular} \label{tab:hawkes_parameters} 
\end{table} 

\begin{table}[ht] \centering \caption{Trading Actions and Execution Variables} \begin{tabular}{ll} \toprule Symbol & Description \\ \midrule $\pi$ & Trading action \\ $l^a$ & Limit sell order size \\ $l^b$ & Limit buy order size \\ $r$ & Exchange rebate \\ $\tau$ & Time of the first subsequent price movement \\ $E$ & Event that the submitted order is executed \\ $E^c$ & Complement of $E$ \\ \bottomrule \end{tabular} \label{tab:trading_controls} 
\end{table} 

\begin{table}[ht] \centering \caption{Standing Assumptions} \begin{tabular}{ll} \toprule Assumption & Description \\ \midrule (A1) & $\mu_i>0$ for all components \\ (A2) & $\Phi(t)\ge 0$ and $\Psi(t)\ge 0$ entrywise, for all $t\ge0$ \\ (A3) & $\int_0^\infty \Phi(s)\,ds < \infty$ and $\int_0^\infty \Psi(s)\,ds < \infty$, entrywise \\   \bottomrule \end{tabular} \label{tab:assumptions} \end{table}

\newpage
\section{Model Buildup: Hawkes-Flocking Limit Order Book}

\subsection{Limit Order Book Dynamics}

We model the dynamics of the best ask and best bid prices using a
multivariate point process framework. Let $A(t)$ and $B(t)$ denote the
best ask price and best bid price at time $t$, respectively.

The arrivals of bid and ask price movements are described by four
counting processes,

\begin{equation}
\mathbb{N}_t
=
\begin{bmatrix}
N_A^u(t) \\
N_A^d(t) \\
N_B^u(t) \\
N_B^d(t)
\end{bmatrix},
\end{equation}

where

\begin{itemize}
\item $N_A^u(t)$ counts upward movements of the best ask price;
\item $N_A^d(t)$ counts downward movements of the best ask price;
\item $N_B^u(t)$ counts upward movements of the best bid price;
\item $N_B^d(t)$ counts downward movements of the best bid price.
\end{itemize}

Assuming that prices move in multiples of a fixed tick size $\delta$,
the best ask and best bid prices can be represented as

\begin{equation}
A(t)
=
A(0)
+
\delta
\Bigl(
N_A^u(t)-N_A^d(t)
\Bigr),
\end{equation}

and

\begin{equation}
B(t)
=
B(0)
+
\delta
\Bigl(
N_B^u(t)-N_B^d(t)
\Bigr).
\end{equation}

The bid--ask spread is defined by

\begin{equation}
S(t)=A(t)-B(t).
\end{equation}

Substituting the price dynamics into the definition of the spread gives

\begin{equation}
S(t)
=
S(0)
+
\delta\Bigl(
N_A^u(t)
+
N_B^d(t)
-
N_A^d(t)
-
N_B^u(t)
\Bigr).
\end{equation}

The spread therefore evolves through two distinct classes of events.
The events $A^u$ and $B^d$ widen the spread, whereas the events $A^d$
and $B^u$ narrow the spread.

To emphasize this distinction, define the spread-widening and
spread-narrowing counting processes by

\begin{align}
N_S^u(t)
&=
N_A^u(t)+N_B^d(t),
\\
N_S^d(t)
&=
N_A^d(t)+N_B^u(t).
\end{align}

Then the spread process may be written as

\begin{equation}
S(t)
=
S(0)
+
\delta
\Bigl(
N_S^u(t)-N_S^d(t)
\Bigr).
\end{equation}

The limit order book imposes a natural constraint on the spread.
Because the best ask price must always remain above the best bid price,
and prices are quoted on a discrete tick grid, the spread can never fall
below one tick. Accordingly, the admissible state space is

\begin{equation}
S(t)\ge \delta.
\end{equation}

This constraint plays an important role in the subsequent model construction. Whenever the spread reaches its minimum level, spread-narrowing price movements must be suppressed to preserve
the ordering of the best bid and ask prices.

The spread process also serves as the key state variable in the Hawkes-flocking specification developed below. The activation of certain intensity components depends on whether the spread is equal to its minimum value $\delta$ or strictly greater than $\delta$.

\subsection{Hawkes-Flocking Structure}

We now specify the intensity process governing the arrivals of bid and
ask price movements.

Let

\begin{equation}
\bm{\lambda}_t
=
\begin{bmatrix}
\lambda_A^u(t)
\\
\lambda_A^d(t)
\\
\lambda_B^u(t)
\\
\lambda_B^d(t)
\end{bmatrix}
\end{equation}

denote the vector of conditional intensities associated with the
counting process $\mathbb N_t$.

The intensity process is defined by

\begin{equation}
\bm{\lambda}_t
=
\left[
\bm{\mu}
+
\int_{-\infty}^{t}
\bm{h}(t-u)\,
d\mathbb N_u
\right]
\circ
\begin{bmatrix}
1
\\
\mathcal I(S(t)>\delta)
\\
\mathcal I(S(t)>\delta)
\\
1
\end{bmatrix},
\label{eq:intensity}
\end{equation}

where

\begin{equation}
\bm{\mu}
=
\begin{bmatrix}
\mu_1
\\
\mu_1
\\
\mu_2
\\
\mu_2
\end{bmatrix}
\end{equation}

is the vector of baseline intensities.

The indicator structure in Equation~(\ref{eq:intensity}) imposes the spread constraint introduced
in the previous subsection. Whenever the spread reaches its minimum
level $\delta$, the spread-narrowing intensities
$\lambda_A^d(t)$ and $\lambda_B^u(t)$ are forced to be zero.
Consequently, no further spread-narrowing movements can occur when
$S(t)=\delta$, ensuring that the spread never falls below one tick.

The kernel matrix is decomposed into a Hawkes component and a flocking
component,

\begin{equation}
\bm{h}(t-u)
=
\Phi(t-u)
+
k(t)\circ\Psi(t-u).
\label{eq:hawkesflockingkernel}
\end{equation}

The matrix $\Phi$ captures self-excitation and mutual-excitation within
the ask and bid price processes,

\begin{equation}
\Phi(t)
=
\begin{bmatrix}
\alpha_{1s}e^{-\beta_1 t}
&
\alpha_{1c}e^{-\beta_1 t}
&
0
&
0
\\
\alpha_{1c}e^{-\beta_1 t}
&
\alpha_{1s}e^{-\beta_1 t}
&
0
&
0
\\
0
&
0
&
\alpha_{2s}e^{-\beta_2 t}
&
\alpha_{2c}e^{-\beta_2 t}
\\
0
&
0
&
\alpha_{2c}e^{-\beta_2 t}
&
\alpha_{2s}e^{-\beta_2 t}
\end{bmatrix}.
\label{eq:phi}
\end{equation}

The parameters $\alpha_{is}$ describe self-excitation, while the
parameters $\alpha_{ic}$ describe mutual-excitation. The decay
parameters $\beta_i$ determine the speed at which the excitation
returns toward the baseline intensity level.

To account for cross-side excitation that depends on the prevailing spread, we introduce a flocking component (the name follows \citet{jang2020systemic}). The activation matrix is defined by

\begin{equation}
k(t)
=
\begin{bmatrix}
\mathbbm{1}_{\{S(t)=\delta\}}
&
\mathbbm{1}_{\{S(t)=\delta\}}
&
\mathbbm{1}_{\{S(t)=\delta\}}
&
\mathbbm{1}_{\{S(t)=\delta\}}
\\
\mathbbm{1}_{\{S(t)>\delta\}}
&
\mathbbm{1}_{\{S(t)>\delta\}}
&
\mathbbm{1}_{\{S(t)>\delta\}}
&
\mathbbm{1}_{\{S(t)>\delta\}}
\\
\mathbbm{1}_{\{S(t)>\delta\}}
&
\mathbbm{1}_{\{S(t)>\delta\}}
&
\mathbbm{1}_{\{S(t)>\delta\}}
&
\mathbbm{1}_{\{S(t)>\delta\}}
\\
\mathbbm{1}_{\{S(t)=\delta\}}
&
\mathbbm{1}_{\{S(t)=\delta\}}
&
\mathbbm{1}_{\{S(t)=\delta\}}
&
\mathbbm{1}_{\{S(t)=\delta\}}
\end{bmatrix}.
\label{eq:k}
\end{equation}

The flocking kernel is given by

\begin{equation}
\Psi(t)
=
\begin{bmatrix}
0
&
0
&
\alpha_{1n}e^{-\beta_1 t}
&
\alpha_{1w}e^{-\beta_1 t}
\\
0
&
0
&
\alpha_{1n}e^{-\beta_1 t}
&
\alpha_{1w}e^{-\beta_1 t}
\\
\alpha_{2w}e^{-\beta_2 t}
&
\alpha_{2n}e^{-\beta_2 t}
&
0
&
0
\\
\alpha_{2w}e^{-\beta_2 t}
&
\alpha_{2n}e^{-\beta_2 t}
&
0
&
0
\end{bmatrix}.
\label{eq:psi}
\end{equation}

The parameters $\alpha_{iw}$ and $\alpha_{in}$ capture the flocking
effects associated with spread-widening and spread-narrowing states,
respectively.

The flocking component differs from the standard Hawkes kernel. The Hawkes kernel models excitation generated by previous movements on the same side of the book. The flocking kernel links the two sides: movements on one side excite movements on the other, and the current spread determines which of these interactions are active. The name is inherited from \citet{jang2020systemic}; the dependence is event-triggered, unlike the persistent synchronization between buy and sell arrivals modeled by \citet{choi2021optimal}. In the estimates of Section~3.3, its fitted role is a rapid, millisecond-scale re-narrowing of the spread after a widening event. The matrix $k(t)$ determines when each type of flocking interaction becomes active.

In sum, the proposed specification combines endogenous clustering
through Hawkes excitation with state-dependent interactions generated by
the spread process. This additional source of dependence plays an
important role in both estimation and the trading framework developed in
subsequent sections.

\begin{remark}[Comparison with the symmetric flocking specification]
In the two-price flocking model of \citet{jang2020systemic}, the roles of $\alpha_{iw}$ and $\alpha_{in}$
are swapped between the `upward' and `downward' rows of the flocking kernel. This is necessary
because their two prices $C_1$ and $C_2$ are symmetric and interchangeable, so whether an event widens
or narrows the price difference depends on the current regime (i.e., whether $C_1<C_2$ or $C_1>C_2$).
In our setting, by contrast, $A(t)$ and $B(t)$ play fixed, asymmetric roles as the best ask and best bid
prices, and $S(t)=A(t)-B(t)$ is a signed quantity with a fixed orientation. Consequently, whether an
event widens or narrows the spread is an absolute property of the event itself, not something that depends on the current regime: $A^u$ and $B^d$ always widen the spread, while $A^d$ and $B^u$ always narrow it. This is why the flocking coefficients $\alpha_{iw}, \alpha_{in}$ in
$\Psi(t)$ do not need to swap between rows: each column's coefficient reflects the fixed
widening/narrowing identity of its triggering event, so the row pairs $(A^u, A^d)$ and $(B^u, B^d)$
share identical flocking responses to a given triggering event, with the regime-dependence handled
entirely by the activation matrix $k(t)$.
\end{remark}

\subsection{Relation to Existing Models and Non-explosion}

The proposed Hawkes-flocking limit order book model combines three
components: a multivariate Hawkes process, a constrained limit order
book structure, and a state-dependent flocking mechanism.

When the flocking kernel is removed, that is,

\begin{equation}
\Psi(t)\equiv 0,
\end{equation}

the model reduces to a constrained Hawkes model for bid and ask price
dynamics. In this case, all dependence among future price movements is
generated through the self-exciting and mutually-exciting structure
contained in the Hawkes kernel $\Phi$.

If the spread constraint is removed, the model has a Hawkes-flocking structure with state-dependent interactions. We do not claim that it coincides with any particular published Hawkes-flocking specification; the present framework builds on, and adapts, both constrained Hawkes models for limit order books and Hawkes-flocking models for interacting price processes.

The two kernels play different roles. The kernel $\Phi$ captures
dependence generated by previous bid and ask price movements, whereas
the kernel $\Psi$ captures cross-side excitation whose activation depends on the spread. The activation matrix $k(t)$ determines whether a
particular flocking interaction becomes active according to the current
state of the spread process.

Several other Hawkes-type models let the dynamics depend on the state of the book, and they differ in how the state enters. In the state-dependent Hawkes process of \citet{morariupatrichi2022state}, the kernel depends on the state at the time of the exciting event, whereas here the activation matrix is evaluated at the current time and switches the contribution of all past events on or off. \citet{sfendourakis2020lob} and \citet{ruan2023selfexciting} multiply the intensity by a function of the current state: \citet{sfendourakis2020lob} use smooth exponential factors, including a one-tick indicator covariate for large-tick stocks, while in \citet{ruan2023selfexciting} the factor for downward spread jumps is zero at one tick, as in the indicator of \citet{zheng2014modelling}. The models also differ in how intensities are kept valid. \citet{lee2023modeling} enforce a nonnegative intensity through a stochastic reset of the excitation and a spread-dependent baseline, and \citet{jain2024limit} floor inhibitory kernels at zero and let the in-spread intensity vanish at one tick. Here the intensity is a nonnegative combination of nonnegative kernels and a $0$--$1$ gate, which is what allows the domination argument below.

To study the non-explosion of the model, we first observe that every
entry of the activation matrix satisfies

\begin{equation}
k_{ij}(t)\in\{0,1\}.
\end{equation}

Furthermore, by construction,

\begin{equation}
\alpha_{iw}\ge0,
\qquad
\alpha_{in}\ge0,
\qquad
i=1,2,
\end{equation}

which implies that every entry of the flocking kernel $\Psi(t)$ is
nonnegative.

Therefore,

\begin{equation}
k(t)\circ\Psi(t)
\le
\Psi(t)
\end{equation}

componentwise. Consequently,

\begin{equation}
h(t)
=
\Phi(t)
+
k(t)\circ\Psi(t)
\le
\Phi(t)
+
\Psi(t).
\end{equation}

Define the dominating kernel

\begin{equation}
\bar h(t)
=
\Phi(t)
+
\Psi(t),
\end{equation}

and define the associated branching matrix

\begin{equation}
\bar M
=
\int_0^\infty
\bar h(u)\,du.
\end{equation}

The entries of $\bar M$ represent the average numbers of offspring
events generated by a single event through both Hawkes excitation and
flocking excitation. The spectral radius of $\bar M$, denoted by $\rho(\bar M)$, is defined as the largest absolute value among the eigenvalues of $\bar M$. From the branching-process interpretation of Hawkes processes, $\rho(\bar M)$ summarizes the overall degree of endogenous amplification present in the system. It plays no role in the non-explosion result below: for $\rho(\bar M)<1$ the dominating process is subcritical, with finite mean intensity, while for $\rho(\bar M)\ge1$ it is supercritical but still non-explosive on every finite horizon.

\begin{lemma}[Monotone thinning against a dominating Hawkes process]\label{lem:coupling}
Let $\bar{\mathbb N}$ be a multivariate Hawkes process with baseline $\bm\mu\ge0$ and a
time-invariant, entrywise-nonnegative, bounded, entrywise-integrable kernel $\bar h(\cdot)$, with $\bar M=\int_0^\infty \bar h(u)\,du$. Existence of a version of $\bar{\mathbb N}$ and its almost-sure non-explosion on every finite horizon follow from \citet[Theorem~2.4(i)]{morariupatrichi2022state} applied with a one-point state space (for $\rho(\bar M)<1$ they also follow from \citet{bremaud1996stability}).

Let $\bm g:[0,\infty)\to[0,1]^4$ be any predictable, entrywise $[0,1]$-valued process, and let
$h(t-u)$ be any entrywise-nonnegative, predictable kernel satisfying
\[
h(t-u)\;\le\;\bar h(t-u)\qquad\text{entrywise, for all }t\ge u,\text{ a.s.}
\]
Then there exists a probability space carrying $\bar{\mathbb N}$ together with a point process
$\mathbb N$ having $\mathcal F_t$-intensity
\[
\bm\lambda_t
=
\Bigl[
\bm\mu+\int_{-\infty}^{t} h(t-u)\,d\mathbb N_u
\Bigr]
\circ
\bm g(t),
\]
such that, almost surely,
\[
\mathbb N_t \;\le\; \bar{\mathbb N}_t \qquad\text{entrywise, for every } t\ge0,
\]
in the sense that $\mathbb N([0,t])\le \bar{\mathbb N}([0,t])$ componentwise as counting measures.
In particular, $\mathbb N$ exists and is almost surely finite on every finite time horizon.
\end{lemma}

\begin{proof}
List the points of $\bar{\mathbb N}$ in increasing time order as $(t_n,i_n)_{n\ge1}$, where
$i_n\in\{A^u,A^d,B^u,B^d\}$ denotes the type of the $n$-th arrival. Let $(U_n)_{n\ge1}$ be an
i.i.d.\ sequence of $\mathrm{Uniform}(0,1)$ random variables, independent of $\bar{\mathbb N}$,
and let $\mathcal F_t$ be the filtration generated by $\bar{\mathbb N}$, $\mathbb N$, and the
marks $(U_n)$ up to time $t$.

We construct $\mathbb N$ by thinning the points of $\bar{\mathbb N}$, processed in time order,
and show inductively that $\mathbb N_s\le\bar{\mathbb N}_s$ entrywise for all $s<t_n$ at each
step $n$. The base case $\mathbb N_0=\bar{\mathbb N}_0=0$ is immediate.

\emph{Inductive step.} Suppose $\mathbb N_s\le\bar{\mathbb N}_s$ entrywise for all $s<t_n$.
Define
\[
\lambda_{t_n}^{(i_n)}
=
\Bigl[
\mu_{i_n}+\int_{-\infty}^{t_n} h_{i_n,\cdot}(t_n-u)\,d\mathbb N_u
\Bigr]
g_{i_n}(t_n),
\qquad
\bar\lambda_{t_n}^{(i_n)}
=
\mu_{i_n}+\int_{-\infty}^{t_n} \bar h_{i_n,\cdot}(t_n-u)\,d\bar{\mathbb N}_u.
\]
Since $d\mathbb N_u\le d\bar{\mathbb N}_u$ as measures on $(-\infty,t_n)$ by the inductive
hypothesis, $h\le \bar h$ entrywise, and $g_{i_n}(t_n)\in[0,1]$, we obtain
\[
\lambda_{t_n}^{(i_n)}
\;\le\;
\mu_{i_n}+\int_{-\infty}^{t_n} \bar h_{i_n,\cdot}(t_n-u)\,d\bar{\mathbb N}_u
\;=\;
\bar\lambda_{t_n}^{(i_n)}.
\]
If $\bar\lambda_{t_n}^{(i_n)}=0$ then also $\lambda_{t_n}^{(i_n)}=0$ and we reject the candidate
point. Otherwise, accept the point into $\mathbb N$ (i.e., set
$\mathbb N_{t_n}=\mathbb N_{t_n^-}+e_{i_n}$) if
$U_n\le \lambda_{t_n}^{(i_n)}/\bar\lambda_{t_n}^{(i_n)}$, and reject it otherwise.

By construction $\mathbb N$ gains a point at time $t_n$ only if $\bar{\mathbb N}$ does, so
$\mathbb N_{t_n}\le\bar{\mathbb N}_{t_n}$ entrywise, completing the induction; since jumps occur
only at the $t_n$, $\mathbb N_t\le\bar{\mathbb N}_t$ for all $t\ge0$.

It remains to verify that the accepted process $\mathbb N$ has $\mathcal F_t$-intensity exactly
$\bm\lambda_t$ as claimed. This is the standard thinning representation of a point process with
stochastic intensity against a dominating point process using independent uniform marks (see,
e.g., \citep[Sec.~7.5]{daleyverejones2003}): since $(U_n)$ is independent
of $\bar{\mathbb N}$ and $\mathbb N$ is a measurable function of $\bar{\mathbb N}$ and $(U_n)$
alone, the acceptance step at each candidate point is a Bernoulli thinning with
$\mathcal F_{t_n^-}$-measurable acceptance probability $\lambda_{t_n}^{(i_n)}/\bar\lambda_{t_n}^{(i_n)}$,
which yields an accepted process with $\mathcal F_t$-intensity $\bm\lambda_t$ by the same argument
used to justify Ogata's thinning algorithm~\cite{ogata1981lewis}.

Finally, $\bar{\mathbb N}_t<\infty$ a.s.\ for every finite $t$ by non-explosion of $\bar{\mathbb N}$,
so $\mathbb N_t\le\bar{\mathbb N}_t<\infty$ a.s.\ as well.
\end{proof}

Lemma~\ref{lem:coupling} is stated for a general dominated kernel $h$ and gate $\bm g$, so
that it applies to any predictable state-dependent modification of a Hawkes intensity, not
only the specific spread-gated structure of Equation~(\ref{eq:intensity}). We now specialize
it to establish the non-explosion of the Hawkes-flocking limit order book model.

\begin{proposition}[Non-explosion of the Hawkes-Flocking Limit Order Book Model]\label{prop:wellposed}
Let the kernels $\Phi$ and $\Psi$ be as in Section~2.2, with nonnegative coefficients $\alpha$ and positive decay rates $\beta_i$. Then the multivariate Hawkes process associated with the dominating kernel $\bar h$ is non-explosive; no condition on $\rho(\bar M)$ is required.

Consequently, the constrained Hawkes-flocking limit order book process is also non-explosive. In particular, almost surely, only finitely many bid and ask price movements occur on every finite time interval.
\end{proposition}

\begin{proof}
Recall $h(t)\le\bar h(t)$ entrywise, as established above. Every entry of $\bar h$ is a finite sum of terms $\alpha e^{-\beta t}$ with $\alpha\ge0$ and $\beta>0$, so $\bar h$ is nonnegative, bounded and integrable, and Lemma~\ref{lem:coupling} applies: there exists a multivariate Hawkes process $\bar{\mathbb N}$ with baseline $\bm\mu$ and kernel $\bar h$, non-explosive on every finite horizon (\citet[Theorem~2.4(i)]{morariupatrichi2022state}).

The intensity process of the constrained model, Equation~(\ref{eq:intensity}), is of the form
\[
\bm\lambda_t
=
\Bigl[
\bm\mu+\int_{-\infty}^t h(t-u)\,d\mathbb N_u
\Bigr]
\circ
\bm g(t),
\qquad
\bm g(t)=\bigl(1,\ \mathcal I(S(t)>\delta),\ \mathcal I(S(t)>\delta),\ 1\bigr)^\top,
\]
where $\bm g(t)$ is entrywise $\{0,1\}$-valued, hence $[0,1]$-valued, and predictable (it depends
only on $S(t^-)$, a function of the process's strict past). Since $h\le\bar h$ entrywise, Lemma~\ref{lem:coupling} applies with this $h$ and $\bm g$, and yields a version of the constrained
Hawkes-flocking process $\mathbb N$, coupled with $\bar{\mathbb N}$ on a common probability
space, satisfying
\[
\mathbb N_t\;\le\;\bar{\mathbb N}_t\qquad\text{entrywise, a.s., for every }t\ge0.
\]
This establishes the existence of $\mathbb N$. Since $\bar{\mathbb N}_t<\infty$ a.s.\ for
every finite $t$, so is $\mathbb N_t$; that is, almost surely, only finitely many bid and ask price
movements occur on every finite time interval, establishing non-explosion of the constrained
Hawkes-flocking limit order book process.
\end{proof}

The proposition establishes that the model is mathematically well
defined on finite horizons. This property is important for likelihood
based parameter estimation, simulation of sample paths, and the trading
framework developed in later sections. The existence of a non-explosive version also guarantees that the conditional intensity process
remains suitable for forecasting future bid and ask movements and for
evaluating trading decisions.

\begin{remark}[Non-explosion is not recurrence]
Proposition~\ref{prop:wellposed} is a finite-horizon statement. The spectral radius $\rho(\bar M)$ enters only through the growth of the dominating process: for $\rho(\bar M)<1$ it is subcritical, and otherwise its expected number of events grows exponentially in $T$ while remaining finite. The proposition does not establish recurrence of the spread or a stationary distribution of the joint spread--intensity system. Such properties have been proved for simple state-dependent spread models: \citet{ruan2023selfexciting} show ergodicity for one-tick jumps and a single exponential kernel, using an intensity for downward jumps that grows with the spread. By contrast, \citet{sfendourakis2020lob} describe the stability of Hawkes processes with a state-dependent factor as open. In our fitted models the total compensator over the estimation window implies more spread-widening than spread-narrowing events (Section~3.3), so whether the fitted spread process is recurrent is an empirical question that we do not settle.
\end{remark}

\section{Estimation of Parameters}

This section describes the estimation method used for the Hawkes-flocking model. The log-likelihood of the Hawkes model and its MLE theory have been investigated by {Ogata}~\cite{ogata1978asymptotic} and {Ozaki}~\cite{ozaki1979maximum}. Other estimation approaches include the conditional least-squares method of {Kirchner}~\cite{kirchner2017estimation} and the non-parametric method of {Bacry, Dayri, and Muzy}~\cite{bacry2012non}. We adopt the maximum likelihood approach, following the estimation method for the Hawkes-flocking model introduced by {Jang, Lee, and Lee}~\cite{jang2020systemic}. Maximum likelihood is a natural choice here: unlike the non-parametric approach, it targets the specific parametric family $(\mu_i,\beta_i,\alpha_{is},\alpha_{ic},\alpha_{in},\alpha_{iw})$ of Section~2, whose entries have direct economic interpretations (self-, mutual-, and flocking-excitation); and unlike conditional least-squares, it extends naturally to the state-dependent gating structure of the intensity process, Equation~\eqref{eq:intensity}, without modification.

Likelihood-based inference for a point process implicitly presumes that the process is well defined and non-explosive on the estimation horizon $[0,T]$. This is guaranteed by Proposition~\ref{prop:wellposed} for every parameter set in the model's domain, so the log-likelihood below is a well-defined, finite random variable. Asymptotic properties of the maximum likelihood estimator would additionally require stationarity and ergodicity, which we do not establish (see the remark after Proposition~\ref{prop:wellposed}). We therefore report asymptotic standard errors only as approximate, descriptive measures of precision, and base inference on the simulation study of Section~3.2 and on likelihood-ratio and information-criterion comparisons of relative fit. We report $\rho(\bar M)$, and regime-specific analogues, as descriptive measures of endogeneity (Section~3.3).

The log-likelihood function up to time $T$ is

\begin{equation}
\begin{aligned}
 L(\theta)=&\sum_{j=1}^{N_A^u (T)} \log \lambda_A^u (t_{A,j}^u) +\sum_{j=1}^{N_A^d (T)} \log \lambda_A^d (t_{A,j}^d) +\sum_{j=1}^{N_B^u (T)} \log \lambda_B^u (t_{B,j}^u) +\sum_{j=1}^{N_B^d (T)} \log \lambda_B^d (t_{B,j}^d) \\
&- \int_{0}^{T}  \lambda_A^u(s) + \lambda_A^d(s)+\lambda_B^u(s)+ \lambda_B^d(s) \, ds,
\end{aligned}
\label{eq:loglik}
\end{equation}

where $\bm\lambda_t$ is as defined in Equation~\eqref{eq:intensity} and $t_{i,j}^{k}$ denotes the associated event times. The parameter set $\theta=(\mu_i,\beta_i,\alpha_{is},\alpha_{ic},\alpha_{in},\alpha_{iw})$ for $i=1,2$ (12 parameters in total) is estimated by maximizing $L(\theta)$ numerically.

\subsection{Computational Remarks}

Direct evaluation of Equation~\eqref{eq:loglik} from its definition requires summing over all past event times at every point of evaluation, an $O(N^2)$ computation in the total number of events $N$. Because every entry of $\Phi$ and $\Psi$ in a given row $i$ shares the same decay rate $\beta_i$ (Equations~\eqref{eq:phi} and \eqref{eq:psi}), the kernel contribution to $\lambda_i$ can instead be tracked as a single exponentially-decaying accumulator that is updated once per event, in the same spirit as the recursive evaluation used for the classical Hawkes log-likelihood (see, e.g., {Ozaki}~\cite{ozaki1979maximum}). The one modification needed relative to the classical recursion is that the activation matrix $k(t)$ is evaluated at the \emph{current} time rather than at the past event time, so the accumulator's coefficient must be re-weighted whenever the spread regime ($S(t)=\delta$ versus $S(t)>\delta$) changes; since $S(t)$ is itself piecewise constant between events, this re-weighting occurs naturally at each event time and does not increase the asymptotic cost. The resulting evaluation is $O(N)$, which is used throughout the simulation study and empirical estimation below.

\subsection{Simulation Validation}

In this subsection, we verify that the maximum likelihood estimator recovers known parameters from data simulated under the model itself, following the same validation logic as {Jang, Lee, and Lee}~\cite{jang2020systemic}. Sample paths are simulated on $[0,T]$ using a multivariate extension of Ogata's modified thinning algorithm~\cite{ogata1981lewis}, applied to the intensity process of Equation~\eqref{eq:intensity} with a parameter set satisfying $\rho(\bar M)<1$, so that the dominating process is subcritical. For each of $R$ independent replications, $\theta$ is re-estimated by maximizing Equation~\eqref{eq:loglik}, and the bias, standard deviation, and root-mean-squared error of the resulting estimates, relative to the true values, are reported in Table~\ref{table:simvalidation} across three estimation horizons $T$.

\begin{table}[htbp]
\centering
\begin{tabular}{l | c | c | c | c}
$\theta$ & True & Bias & Std. & RMSE \\
\hline
\multicolumn{5}{c}{$T=500$ (334.6 avg.\ events/path)} \\
\hline
$\mu_1$ & 0.0800 & 0.0037 & 0.0157 & 0.0160 \\
$\mu_2$ & 0.0800 & -0.0004 & 0.0137 & 0.0136 \\
$\beta_1$ & 0.6000 & 0.0211 & 0.1381 & 0.1393 \\
$\beta_2$ & 1.2000 & 0.1106 & 0.5120 & 0.5226 \\
$\alpha_{1s}$ & 0.2400 & -0.0060 & 0.0583 & 0.0584 \\
$\alpha_{1c}$ & 0.0600 & 0.0018 & 0.0362 & 0.0361 \\
$\alpha_{1n}$ & 0.3000 & 0.0141 & 0.1005 & 0.1013 \\
$\alpha_{1w}$ & 0.1200 & -0.0031 & 0.0691 & 0.0690 \\
$\alpha_{2s}$ & 0.2400 & -0.0061 & 0.0996 & 0.0995 \\
$\alpha_{2c}$ & 0.0600 & 0.0104 & 0.0574 & 0.0582 \\
$\alpha_{2n}$ & 0.3000 & 0.0178 & 0.1340 & 0.1349 \\
$\alpha_{2w}$ & 0.1200 & 0.0159 & 0.0861 & 0.0873 \\
\hline
\multicolumn{5}{c}{$T=2000$ (1344.9 avg.\ events/path)} \\
\hline
$\mu_1$ & 0.0800 & 0.0001 & 0.0087 & 0.0087 \\
$\mu_2$ & 0.0800 & 0.0013 & 0.0066 & 0.0067 \\
$\beta_1$ & 0.6000 & 0.0084 & 0.0646 & 0.0650 \\
$\beta_2$ & 1.2000 & 0.0230 & 0.1793 & 0.1803 \\
$\alpha_{1s}$ & 0.2400 & 0.0018 & 0.0297 & 0.0297 \\
$\alpha_{1c}$ & 0.0600 & -0.0007 & 0.0194 & 0.0194 \\
$\alpha_{1n}$ & 0.3000 & 0.0039 & 0.0459 & 0.0459 \\
$\alpha_{1w}$ & 0.1200 & 0.0029 & 0.0327 & 0.0327 \\
$\alpha_{2s}$ & 0.2400 & -0.0016 & 0.0433 & 0.0432 \\
$\alpha_{2c}$ & 0.0600 & -0.0048 & 0.0261 & 0.0265 \\
$\alpha_{2n}$ & 0.3000 & 0.0098 & 0.0629 & 0.0635 \\
$\alpha_{2w}$ & 0.1200 & 0.0014 & 0.0408 & 0.0408 \\
\hline
\multicolumn{5}{c}{$T=10000$ (6734.7 avg.\ events/path)} \\
\hline
$\mu_1$ & 0.0800 & 0.0005 & 0.0036 & 0.0036 \\
$\mu_2$ & 0.0800 & 0.0006 & 0.0030 & 0.0030 \\
$\beta_1$ & 0.6000 & 0.0048 & 0.0268 & 0.0271 \\
$\beta_2$ & 1.2000 & 0.0059 & 0.0884 & 0.0884 \\
$\alpha_{1s}$ & 0.2400 & -0.0001 & 0.0121 & 0.0121 \\
$\alpha_{1c}$ & 0.0600 & 0.0007 & 0.0080 & 0.0080 \\
$\alpha_{1n}$ & 0.3000 & 0.0009 & 0.0229 & 0.0228 \\
$\alpha_{1w}$ & 0.1200 & 0.0012 & 0.0156 & 0.0156 \\
$\alpha_{2s}$ & 0.2400 & -0.0027 & 0.0221 & 0.0222 \\
$\alpha_{2c}$ & 0.0600 & -0.0009 & 0.0134 & 0.0134 \\
$\alpha_{2n}$ & 0.3000 & 0.0012 & 0.0280 & 0.0279 \\
$\alpha_{2w}$ & 0.1200 & 0.0009 & 0.0188 & 0.0188 \\
\end{tabular}
\caption{Simulation-based validation of the ML estimator across three estimation horizons: true parameter values (satisfying $\rho(\bar M)\approx0.886<1$) versus the bias, standard deviation, and root-mean-squared error of the ML estimate across $R=200$ independent sample paths simulated on $[0,T]$, for $T\in\{500,2000,10000\}$.}
\label{table:simvalidation}
\end{table}

The baseline intensities $\mu_i$, decay rate $\beta_1$, and self-exciting parameters $\alpha_{is}$ recover cleanly at every horizon, with standard deviations substantially smaller than the parameter values themselves already at $T=500$. The mutual-exciting and flocking parameters ($\alpha_{ic}$, $\alpha_{in}$, $\alpha_{iw}$) and $\beta_2$ show larger standard deviations at $T=500$, consistent with a known identifiability difficulty in multivariate Hawkes estimation. Separating a process's response to \emph{another} process's history (mutual excitation, flocking) from its own self-exciting dynamics is intrinsically harder than identifying self-excitation alone, particularly when the two effects are estimated jointly from a single sample path. {Jang, Lee, and Lee}~\cite{jang2020systemic} document a related phenomenon for the Hawkes-flocking specification, where estimates of $\alpha_s$ and $\alpha_c$ are shown to be affected by near-multicollinearity between the self-exciting and flocking components of the kernel.

Table~\ref{table:simvalidation} shows that this behavior is a finite-sample effect rather than a structural limitation of the estimator. As the horizon lengthens from $T=500$ to $T=10000$, the standard deviation of every parameter shrinks monotonically, at a rate consistent with the $\sqrt{T}$ convergence of standard point-process maximum likelihood asymptotics. The bias in $\beta_2$, the parameter most affected at $T=500$, falls from $0.1106$ ($\approx 9\%$ of its true value) to $0.0230$ at $T=2000$ and $0.0059$ at $T=10000$, becoming statistically indistinguishable from zero at the longer horizons. These results are in line with consistency of the estimator for the full parameter set, including the mutual-exciting and flocking components: the residual bias visible at short horizons appears attributable to sample size, not to the estimation procedure or the model specification, though a simulation study of this kind cannot establish consistency formally.

\subsection{Application to Real Data}\label{sec:realdata}

We apply the estimation procedure of Sections~3.1--3.2 to real intraday limit order book data, using the model to test whether the flocking mechanism $\Psi$ is empirically justified, and to check the fitted dynamics for goodness-of-fit.

We use LOBSTER level-5 limit order book data~\cite{huang2011lobster} for a single trading day (2012-06-21). The model's core assumption, that the best ask and bid move in single-tick increments with $\Psi$ gated by whether $S(t)=\delta$, fits large-tick stocks (where the spread is usually at its one-tick minimum) far better than high-priced, small-tick-relative-to-price stocks. We confirmed this on the data before restricting attention to INTC and MSFT: for AMZN, a high-priced stock in the same sample, the spread is at one tick only $0.16\%$ of the time over the 10:00--15:30 window (time-weighted median spread of $12$ ticks), whereas for INTC and MSFT it sits at its one-tick minimum over $99\%$ of the time (though in only about three-quarters of order-book updates, since brief two-tick excursions generate many updates), and ask and bid moves of more than one tick are rare (under $0.5\%$ of moves for either stock). Best ask/bid states are collapsed across tied timestamps and classified into the four event types $A^u,A^d,B^u,B^d$; simultaneous ask-and-bid moves and multi-tick jumps, both well under $1\%$ of transitions for these two stocks, are excluded from the event series, and the sample is restricted to 10:00--15:30 to avoid open/close seasonality (the window used by \citet{lee2023modeling}; \citet{morariupatrichi2022state} and \citet{sfendourakis2020lob} use narrower mid-day windows, and a robustness check on a narrower window is left for future work). This yields $1{,}768$ events for INTC and $2{,}600$ for MSFT, with $\delta=\$0.01$ matching the model's tick size.

For each stock, $\theta$ is re-estimated by maximizing Equation~\eqref{eq:loglik} using a multi-start protocol (log-uniform random initializations, with a derivative-free fallback when the gradient-based optimizer's line search fails to converge), since we found the likelihood surface for this specification to be sensitive to starting values on real data at this timescale. Table~\ref{table:realdata-estimates} reports the best fit found for each stock. Difficulties of this kind are not specific to our data: \citet{lee2023modeling} report that, for large true decay rates, the success of the optimizer depends on the starting value and improves with sample size (they simulate 5{,}000 to 10{,}000 events; our series contain $1{,}768$ and $2{,}600$).

\begin{table}[htbp]
\centering
\begin{tabular}{l | c c | c c}
$\theta$ & \multicolumn{2}{c|}{INTC} & \multicolumn{2}{c}{MSFT} \\
 & est. & s.e. & est. & s.e. \\
\hline
$\mu_1$ & 0.0202 & 0.0010 & 0.0278 & 0.0012 \\
$\mu_2$ & 0.0220 & 0.0011 & 0.0329 & 0.0013 \\
$\beta_1$ & 438.89 & 25.2 & 550.47 & 25.1 \\
$\beta_2$ & 679.43 & 39.0 & 840.01 & 41.8 \\
$\alpha_{1s}$ & 0.052 & 5.9 & 2.920 & 2.9 \\
$\alpha_{1c}$ & 80.23 & 8.8 & 178.03 & 12.9 \\
$\alpha_{1n}$ & 3.073 & 3.1 & 0.0001 & 4.0 \\
$\alpha_{1w}$ & 300.55 & 23.1 & 411.05 & 25.4 \\
$\alpha_{2s}$ & 0.034 & 31.0 & 9.912 & 12.1 \\
$\alpha_{2c}$ & 164.77 & 21.4 & 128.67 & 12.6 \\
$\alpha_{2n}$ & 0.001 & 12.4 & 0.007 & 7.3 \\
$\alpha_{2w}$ & 496.58 & 40.3 & 591.95 & 41.1 \\
\hline
$\rho(\bar M)$ & \multicolumn{2}{c|}{0.87--0.92} & \multicolumn{2}{c}{0.91--1.00} \\
\end{tabular}
\caption{Maximum likelihood estimates, full 12-parameter model, INTC and MSFT (2012-06-21, 10:00--15:30). The $\beta_i$ estimates correspond to a memory on the order of milliseconds, consistent with the sub-10ms clustering present in both event series. The standard errors (s.e.) are asymptotic, from the observed information matrix at the fit, computed in coordinates scaled by each parameter family's typical size, with coefficients below $0.5$ moved to $0.5$ so that finite differences stay inside the parameter domain (Appendix~\ref{appendix:B}). For coefficients at or near zero, they indicate only that these are not distinguishable from zero. $\rho(\bar M)$ is reported as a range across the best-fitting parameter sets found during the multi-start search. It is a descriptive measure of the endogeneity of the ungated dominating process (Section~2.3); regime-specific radii are discussed in the text.}
\label{table:realdata-estimates}
\end{table}

The decay rates $\beta_i$ of several hundred per second, and the excitation coefficients of order $10^2$, are of the same order as the estimates that \citet{lee2023modeling} report for IBM in January 2018 ($\beta$ of roughly $540$--$1{,}100$ per second and excitation coefficients of roughly $35$--$530$ on 17 of the 21 trading days reported, with both an order of magnitude smaller on the remaining four), even though their data are consolidated quotes of a high-priced stock. The spectral radius $\rho(\bar M)$ in Table~\ref{table:realdata-estimates} is that of the ungated dominating matrix and overstates the amplification of the gated process. Two regime-specific branching matrices are more informative. At minimum spread, only the rows of $A^u$ and $B^d$ are active, and they receive both $\bar\Phi$ and $\bar\Psi$. At wider spreads, all four rows are active, with $\bar\Psi$ entering only the rows of $A^d$ and $B^u$. At the best fit of each stock, their spectral radii are about $0.71$ (INTC) and $0.73$ (MSFT) at minimum spread and $0.44$ and $0.49$ at wider spreads, compared with $0.92$ and $0.98$ for $\rho(\bar M)$. This follows the state-dependent spectral radii of \citet{morariupatrichi2022state}, and is consistent with the higher endogeneity at one-tick spreads reported by \citet{sfendourakis2020lob}; branching-ratio matrices of order-book flows have also been estimated nonparametrically by \citet{achab2018analysis}.

On a single day, the standard errors in Table~\ref{table:realdata-estimates} are roughly $5$--$6\%$ of the typical scale of the baselines $\mu_i$ and $5$--$8\%$ of that of the decay rates $\beta_i$ (about $25$--$42$ per second), and $9$--$41$ units for the largest excitation coefficients $\alpha_{ic}$ and $\alpha_{iw}$. The timescale and cross-side magnitudes are therefore estimated with moderate precision, while the small coefficients $\alpha_{is}$ and $\alpha_{in}$ cannot be distinguished from zero. The observed information is well conditioned in scaled coordinates (condition numbers of about $486$ for INTC and $449$ for MSFT), and about $18$ trading days would bring every direction of the parameter space to within $10\%$ of its scale. These standard errors assume a correctly specified, stationary model, which the goodness-of-fit diagnostics below call into question.

We test $H_0:\Psi\equiv 0$ (the four flocking parameters $\alpha_{1n},\alpha_{1w},\alpha_{2n},\alpha_{2w}$ all zero, leaving a constrained Hawkes model with self- and cross-excitation only) against the full model, re-estimating the null model with the same multi-start protocol. In both models the baseline intensity is shared between the widening and narrowing event types on each side, so this comparison does not on its own separate the effect of $\Psi$ from that of any other spread-restoring mechanism; Appendix~\ref{appendix:B} reports a stronger test with separate baselines and confirms the rejection survives it. Because the null sets each flocking parameter to a boundary value of its ($\geq 0$) domain, the standard $\chi^2_4$ reference distribution for the likelihood-ratio statistic is not exact when parameters lie on the boundary of their domain~\cite{self1987asymptotic}; we report it on that basis, together with AIC and BIC, which are unaffected by the boundary issue. Because the goodness-of-fit diagnostics below show that neither model is correctly specified, we read the comparison as a measure of relative fit.

\begin{table}[htbp]
\centering
\begin{tabular}{l | c c | c c}
 & \multicolumn{2}{c|}{INTC} & \multicolumn{2}{c}{MSFT} \\
 & Full & $\Psi=0$ & Full & $\Psi=0$ \\
\hline
Log-likelihood & $-572.5$ & $-5269.4$ & $107.5$ & $-6867.3$ \\
AIC & $1169.0$ & $10554.8$ & $-187.4$ & $13750.6$ \\
BIC & $1234.7$ & $10598.6$ & $-117.1$ & $13797.5$ \\
\end{tabular}
\caption{Full model versus the $\Psi=0$ null, both stocks. For both, the likelihood-ratio statistic exceeds $9{,}300$ on 4 degrees of freedom, and both AIC and BIC favor the full model by several thousand points, far beyond what the four additional parameters could produce by chance. The restriction $\Psi=0$ is rejected, and the cross-side term substantially improves fit for both stocks on this trading day.}
\label{table:realdata-tests}
\end{table}

We additionally check the fitted dynamics using time-rescaled residuals~\cite{ogata1988statistical}: for each event type, the compensator $\Lambda_i(t)=\int_0^t\lambda_i(s)\,ds$ evaluated at that type's own event times should have i.i.d.\ Exponential(1) increments under a correctly-specified model. Figure~\ref{fig:qqplot-gof} compares the resulting QQ-plots for INTC, full model versus $\Psi=0$; MSFT shows the same pattern (Figure~\ref{fig:appendixB-msft-qq}, Appendix~\ref{appendix:B}). Kolmogorov-Smirnov tests formally reject Exponential(1) for every event type under both models ($p<0.05$), so neither model is a fully correct description of event timing at this resolution. This is consistent with the single-exponential kernel being too rigid to capture the true multi-timescale clustering of real order book events. \citet{lee2023modeling} report slightly fatter tails than Exponential(1) in the same diagnostic for a single-exponential kernel, and point to multi-kernel specifications. Within that limitation, the two models differ in an interpretable way. Over the estimation window, spread-widening events ($A^u,B^d$) and spread-narrowing events ($A^d,B^u$) must occur equally often (884 each for INTC, 1{,}300 each for MSFT), and for a correctly-specified model the total compensator of each type would match its observed count. Under $\Psi=0$ the model generates only 116 of INTC's 884 narrowing events (142 of 1{,}300 for MSFT), against 408 (589) under the full model, because without $\Psi$ nothing in the model closes the spread quickly after it widens. Correspondingly, the mean rescaled residuals of $A^d$ and $B^u$ fall to $0.08$--$0.18$ under $\Psi=0$, and their Kolmogorov-Smirnov statistics roughly double. Since two-tick spreads last only milliseconds, the role of the flocking kernel is essentially this rapid snap-back, and it is the narrowing event types that carry it. The full model remains imperfect here: it still overstates widening events (1{,}337 versus 884 for INTC) and understates narrowing ones.

\begin{figure}[htbp]
\centering
\includegraphics[width=0.95\textwidth]{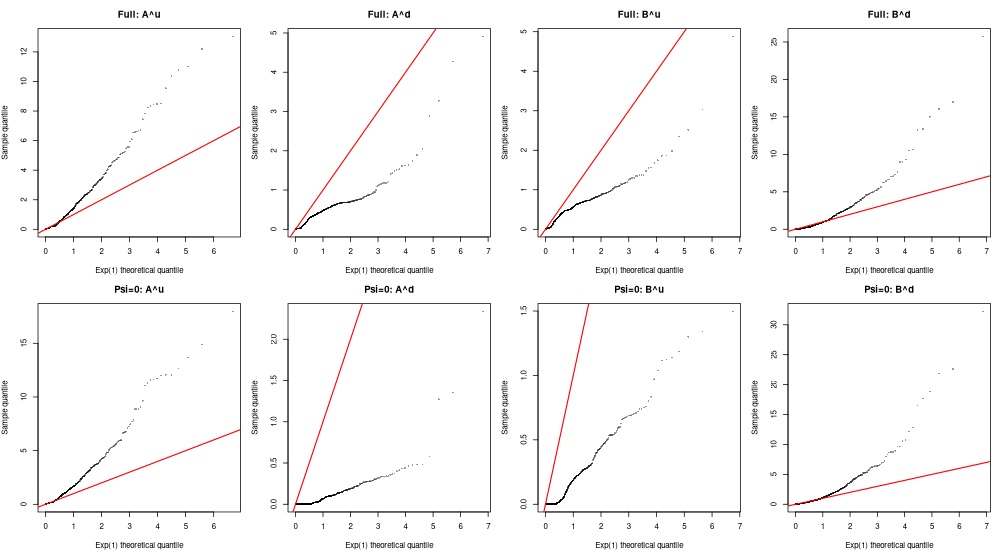}
\caption{Time-rescaled residual QQ-plots against Exponential(1), INTC: full model (top row) versus $\Psi=0$ (bottom row), one panel per event type. The red line is the identity; deviation from it indicates departure from Exponential(1). $A^d$ and $B^u$ compress toward the origin much more severely under $\Psi=0$; $A^u$ and $B^d$ are comparatively stable across both models.}
\label{fig:qqplot-gof}
\end{figure}

\newpage
\section{Cash Flow for Algorithmic Trading}

We now illustrate the model of Section~2 with a single-period trading example. A trader observes the state of the limit order book (the best ask $A(t)$, best bid $B(t)$, and spread $S(t)$), together with the estimated Hawkes-flocking parameters $\theta$, and places a single limit order at time $t$. We evaluate the expected cash-flow impact of that order at the time of the next price movement, $t+\tau$. This single-period formulation isolates the trader's placement decision (price level and quantity) from the sequencing and inventory-management questions that arise when orders are placed repeatedly over a trading horizon; we leave the multi-period extension to future work.

\subsection{Utility Function Incorporating Expected Cash Flow}

Recall the cash position $X(t)$, inventory $Y(t)$, and mark-to-market wealth $G(t):=X(t)+\frac{A(t)+B(t)}{2}Y(t)$ from Table~\ref{tab:state_variables}. For a trading action $\pi(t)$, that is, a choice of order type, price level, and quantity (detailed in Section~\ref{sec:optimal-action} below), define the one-step conditional mean and variance of the resulting change in $G$:
\begin{equation}
\mu_G(t,\pi(t)):=E[G(t+\tau)-G(t)\mid\mathcal{F}_t,\pi(t)], \qquad \sigma^2_G(t,\pi(t)):=Var[G(t+\tau)-G(t)\mid\mathcal{F}_t,\pi(t)].
\end{equation}

We want to maximize the expected gain while controlling its variance, so we introduce a risk-aversion parameter $\eta>0$ and define the utility function $H(t,\pi(t))$ as follows:

\begin{equation}
	H(t,\pi(t))=\mu_G(t,\pi(t))-{\eta}\sigma^2_G(t,\pi(t)),
	\label{def:Ht}
	\end{equation}

and our goal is to maximize $H(t)$.

\subsection{Optimal Action Based on Parameters}\label{sec:optimal-action}

This subsection considers a single time point; the action is chosen once, evaluated at the next price movement $t+\tau$.

At time $t$, the trader may:
\begin{itemize}
	\item sell LO $l^a$ amount at $A(t)$ (best ask), or $A^{+}(t)=A(t)+\delta$ (second best ask), and/or
	\item buy LO $l^b$ amount at $B(t)$ (best bid), or $B^{-}(t)=B(t)-\delta$ (second best bid). 
\end{itemize}

\begin{remark}
These four price levels do not exhaust a trader's options in a real limit order book. An order can also rest deeper than the second-best level, or be placed aggressively inside the spread as a marketable limit order priced to execute immediately. We restrict attention to the best and second-best levels on each side because this section is a single-period illustration of how the model informs a placement decision, not an exhaustive treatment of the trader's action space. Deeper levels and marketable orders are left for future work. (The third-best execution probability $q_3^i$ appears below only because it is needed to characterize what happens to a second-best order if it is bumped one level further; it is not itself offered as a placement choice.) \citet{guo2013optimal} give some support for this restriction: in a different price model (a correlated random walk for the best quotes), their static analysis finds that only the market order and orders at the best and second-best bid matter. We do not claim the same result for the present model. The single-period objective also abstracts from adverse selection, which links an order's fill probability to its profitability~\cite{lehalle2017limit}; in this section the execution probabilities $q_j^i$ are inputs.
\end{remark}

Throughout this section we use the following notation. Let $\lambda_A^u, \lambda_A^d, \lambda_B^u, \lambda_B^d$ denote the conditional intensities from Section~2 evaluated at the current time $t$, abbreviated $\lambda_{A,B}^{u,d}$, and let
\[
\Lambda := \lambda_A^u+\lambda_A^d+\lambda_B^u+\lambda_B^d, \qquad p_{A,B}^{u,d} := \lambda_{A,B}^{u,d}/\Lambda.
\]
Let $q_j^i$ denote the probability that a limit order in the bid ($i=b$) or ask ($i=a$) book, resting at the best ($j=1$), second-best ($j=2$), or third-best ($j=3$) price level, is executed by the time of the next price movement.

\begin{lemma}
When the trader takes action $\pi(t)=(A(t),l^a)$, placing limit sell order at price $A(t)$, the best ask price, of quantity $l^a$, then $\mu_G$ and $\sigma^2_G$ are as follows:

\begin{align*}
\mu_G(t,\pi(t)=(A(t),l^a))
=&l^a \left\{\left( \frac{S(t)}{2}+r \right){(p_A^u +p_A^d q_2^a + (p_B^u + p_B^d )q_1^a) }  +\dfrac{\delta}{2}(p_A^dq_2^a + p_B^d q_1^a -p_A^u - p_B^u q_1^a)\right\}\\
&+\dfrac{\delta}{2} Y(t) (p_A^u + p_B^u - p_B^d -p_A^d) \\
\sigma^2_G(t,\pi(t)=(A(t),l^a))
=&\bigg[(\dfrac{S(t)-\delta}{2}+r)l^a+ \frac{\delta}{2}Y(t)\bigg]^2 p_A^u \\
   +& \bigg[(\dfrac{S(t)+\delta}{2}+r)l^a -  \frac{\delta}{2}Y(t) \bigg]^2 p_A^d q_2^a + (\frac{\delta}{2}Y(t))^2 p_A^d (1-q_2^a)\\
   +&\bigg[(\dfrac{S(t)-\delta}{2}+r)l^a + \frac{\delta}{2}Y(t)  \bigg]^2 p_B^u q_1^a+ (\frac{\delta}{2}Y(t))^2  p_B^u (1-q_1^a)\\
   +&\bigg[(\dfrac{S(t)+\delta}{2}+r)l^a -  \frac{\delta}{2}Y(t) \bigg]^2 p_B^d q_1^a + (\frac{\delta}{2}Y(t))^2  p_B^d (1- q_1^a)
   -\mu_G(t,\pi(t)=(A(t),l^a))^2
\end{align*}

The rest of the cases, when the trader places limit sell order at the second ask price, or when the trader places limit buy order to first/second bid prices, are in Appendix~\ref{appendix:A}.

\label{lem1}
\end{lemma}

\begin{proof}
\par In this proof, let us first consider the case of LO sell placement at the price $A(t)$, with quantity $l^a$. The rest of the proof is provided in Appendix~\ref{appendix:A}.

\par At time $t$, the next price movement could be Ask price up/down ($A^u$/$A^d$), and Bid price up/down  ($B^u$/$B^d$), occurring with conditional intensities $\lambda_{A,B}^{u,d}(t)$ as defined above.

\begin{itemize}
\item If next event is $A^u$, this event implies that the trader's order placed at $A(t)$ has been executed. 
\item If next event is $A^d$, there is a chance of execution, or the trader's order has been moved to the second best ask price. Let us denote the chance of execution as $q_2^a$. Note that $q_2^a$ need not be constant; it may depend on $l^a$ and order flow, and can be computed by adapting the queue-depletion-race approach of \citet{figueroalopez2018optimal} to the present Hawkes-flocking intensity framework.
\item If next event is $B^u$ or $B^d$, there is a chance of execution, or the trader's order is still at the best ask. Let us denote the probability of execution of the trader's order in this case as $q_1^a$.
\end{itemize}

Then, after the first price movement ($t+\tau$), 

\begin{table}[htbp]
\begin{tabular}{c|c|c|c|c}
Event & $X(t+\tau)-X(t)$ & $Y(t+\tau)$ & $A(t+\tau)+B(t+\tau)$ & $G(t+\tau)-G(t)$\\
\hline
$A^u$ & $(A(t)+r) l^a$  & $Y(t)-l^a$  & $A(t)+B(t)+\delta$  & $(\dfrac{S(t)-\delta}{2}+r)l^a + \frac{\delta}{2}Y(t)$\\
\hline
$A^d$, $E$ & $(A(t)+r) l^a$ & $Y(t)-l^a$ & $A(t)+B(t)-\delta$ & $(\dfrac{S(t)+\delta}{2}+r)l^a -  \frac{\delta}{2}Y(t)$\\
$A^d$, $E^c$ & 0  & $Y(t)$  & $A(t)+B(t)-\delta$  & $- \frac{\delta}{2}Y(t)$\\
\hline
$B^u$, $E$ & $(A(t)+r) l^a$ & $Y(t)-l^a$ & $A(t)+B(t)+\delta$ &$(\dfrac{S(t)-\delta}{2}+r)l^a  + \frac{\delta}{2}Y(t)$\\
$B^u$, $E^c$ & 0  & $Y(t)$ & $A(t)+B(t)+\delta$ & $+ \frac{\delta}{2}Y(t)$\\
\hline
$B^d$, $E$ & $(A(t)+r) l^a$ & $Y(t)-l^a$ & $A(t)+B(t)-\delta$ & $(\dfrac{S(t)+\delta}{2}+r)l^a -  \frac{\delta}{2}Y(t)$\\
$B^d$, $E^c$ & 0  & $Y(t)$ & $A(t)+B(t)-\delta$ & $-  \frac{\delta}{2}Y(t)$
\end{tabular}
\caption{event E: order executed}\label{pf:tb1}
\end{table}

Using the results from Table \ref{pf:tb1}, we have

\begin{align*}
    G(t+\tau) - G(t)|_{\pi= (A(t), l^a)} =&\bigg[(\dfrac{S(t)-\delta}{2}+r)l^a+ \frac{\delta}{2}Y(t)\bigg]I (A^u) \\
   +& \bigg[(\dfrac{S(t)+\delta}{2}+r)l^a -  \frac{\delta}{2}Y(t) \bigg] I(A^d, E) -  \frac{\delta}{2}Y(t)I(A^d, E^c)\\
   +&\bigg[(\dfrac{S(t)-\delta}{2}+r)l^a + \frac{\delta}{2}Y(t)  \bigg]I(B^u, E)+ \frac{\delta}{2}Y(t) I(B^u, E^c)\\
   +&\bigg[(\dfrac{S(t)+\delta}{2}+r)l^a-  \frac{\delta}{2}Y(t)  \bigg]I(B^d, E)-  \frac{\delta}{2}Y(t) I(B^d, E^c)\\
   =&\bigg[(\dfrac{S(t)-\delta}{2}+r)l^a \bigg](I (A^u) + I(B^u, E))\\
   +& \bigg[(\dfrac{S(t)+\delta}{2}+r)l^a  \bigg]  ( I(A^d, E) + I(B^d, E))\\
   +& \frac{\delta}{2}Y(t) - {\delta}Y(t) ( I (B^d) + I (A^d))
\end{align*}

\begin{align*}
\mu_G(t,\pi(t)=(A(t),l^a))
=&l^a \left\{\left( \frac{S(t)}{2}+r \right){(p_A^u +p_A^d q_2^a + (p_B^u + p_B^d )q_1^a) }  +\dfrac{\delta}{2}(p_A^dq_2^a + p_B^d q_1^a -p_A^u - p_B^u q_1^a)\right\}\\
&+\dfrac{\delta}{2} Y(t) (p_A^u + p_B^u - p_B^d -p_A^d) \\
\sigma^2_G(t,\pi(t)=(A(t),l^a))
=&\bigg[(\dfrac{S(t)-\delta}{2}+r)l^a+ \frac{\delta}{2}Y(t)\bigg]^2 p_A^u \\
   +& \bigg[(\dfrac{S(t)+\delta}{2}+r)l^a -  \frac{\delta}{2}Y(t) \bigg]^2 p_A^d q_2^a + (\frac{\delta}{2}Y(t))^2 p_A^d (1-q_2^a)\\
   +&\bigg[(\dfrac{S(t)-\delta}{2}+r)l^a + \frac{\delta}{2}Y(t)  \bigg]^2 p_B^u q_1^a+ (\frac{\delta}{2}Y(t))^2  p_B^u (1-q_1^a)\\
   +&\bigg[(\dfrac{S(t)+\delta}{2}+r)l^a -  \frac{\delta}{2}Y(t) \bigg]^2 p_B^d q_1^a + (\frac{\delta}{2}Y(t))^2  p_B^d (1- q_1^a)
   -\mu_G(t,\pi(t)=(A(t),l^a))^2
\end{align*}

Rest of the proof is in Appendix~\ref{appendix:A}. 
\end{proof}

Lemma~\ref{lem1} gives $\mu_G(t,\pi)$ and $\sigma^2_G(t,\pi)$ for a resting order at any of the four price levels. Combined with the utility function $H(t,\pi)=\mu_G(t,\pi)-\eta\sigma^2_G(t,\pi)$ of \eqref{def:Ht}, the trader's placement problem reduces to: for a fixed price level, choose the order quantity $l$ that maximizes $H$. Since each outcome $G(t+\tau)-G(t)$ in the proof of Lemma~\ref{lem1} is affine in $l$, $\mu_G(t,\pi)$ is linear in $l$ and $\sigma^2_G(t,\pi)$ is a quadratic form in $l$, so $H(t,\pi)$ is itself quadratic in $l$ for each fixed price level. The following corollary works this out explicitly for $\pi=(A(t),l^a=l)$; the remaining three price levels follow the same argument.

\begin{corollary}

For action $\pi=(A(t), l^a=l)$ (limit order sell at the best ask), $H(t,\pi)$ is quadratic in $l$:
\begin{equation}
H(t,\pi=(A(t), l^a=l)) = c_2 l^2 + c_1 l + c_0, \label{eq:quadratic-H}
\end{equation}
with coefficients $c_2,c_1,c_0$ given below. We find $l$ which maximizes $H(t,\pi)$ and the maximum value as follows:

\begin{align*}
\arg\max_{l} H(t,\pi=(A(t), l^a=l)) 
=&-\frac{c_1}{2c_2},\label{maxl.a}
\\
\max_{l} H(t,\pi=(A(t), l^a=l))
=&-\frac{c_1^2}{4c_2}+c_0,
\end{align*}

\begin{align*}
c_2=&\eta \left\{\left( \frac{S(t)}{2}+r \right){(p_A^u +p_A^d q_2^a + (p_B^u + p_B^d )q_1^a) }  +\dfrac{\delta}{2}(p_A^dq_2^a + p_B^d q_1^a -p_A^u - p_B^u q_1^a)\right\}^2\\
& -\eta \left\{p_A^u(\dfrac{S(t)-\delta}{2}+r)^2 + p_A^d q_2^a(\dfrac{S(t)+\delta}{2}+r)^2 +  p_B^u q_1^a (\dfrac{S(t)-\delta}{2}+r)^2 +p_B^d q_1^a(\dfrac{S(t)+\delta}{2}+r)^2\right\}
\\
c_1 =&\left\{\left( \frac{S(t)}{2}+r \right){(p_A^u +p_A^d q_2^a + (p_B^u + p_B^d )q_1^a) }  +\dfrac{\delta}{2}(p_A^dq_2^a + p_B^d q_1^a -p_A^u - p_B^u q_1^a)\right\} \\
&+ 2\eta \left\{\left( \frac{S(t)}{2}+r \right){(p_A^u +p_A^d q_2^a + (p_B^u + p_B^d )q_1^a) }  +\dfrac{\delta}{2}(p_A^dq_2^a + p_B^d q_1^a -p_A^u - p_B^u q_1^a)\right\}\dfrac{\delta}{2} Y(t) (p_A^u + p_B^u - p_B^d -p_A^d)\\
  -\eta &\left\{p_A^u(\dfrac{S(t)-\delta}{2}+r)\delta Y(t) - p_A^d q_2^a(\dfrac{S(t)+\delta}{2}+r)\delta Y(t) +  p_B^u q_1^a (\dfrac{S(t)-\delta}{2}+r)\delta Y(t) -p_B^d q_1^a(\dfrac{S(t)+\delta}{2}+r)\delta Y(t)\right\}
\\
c_0 =& \dfrac{\delta}{2} Y(t) (p_A^u + p_B^u - p_B^d -p_A^d) + \eta \dfrac{\delta^2}{4} Y(t)^2 (p_A^u + p_B^u - p_B^d -p_A^d)^2 -\eta\dfrac{\delta^2}{4} Y(t)^2 
\end{align*}

\label{cor:optimal-l}
\end{corollary}

\begin{proof}
By Lemma~\ref{lem1}, $\mu_G(t,\pi)=lK+\frac{\delta}{2}Y(t)(p_A^u+p_B^u-p_B^d-p_A^d)$, where $K$ denotes the coefficient of $l$ given above, and $\sigma^2_G(t,\pi)$ is the variance of a random variable taking the value $\left(\frac{S(t)-\delta}{2}+r\right)l+\frac{\delta}{2}Y(t)$ with probability $p_A^u+p_B^uq_1^a$, the value $\left(\frac{S(t)+\delta}{2}+r\right)l-\frac{\delta}{2}Y(t)$ with probability $p_A^dq_2^a+p_B^dq_1^a$, and values not depending on $l$ on the remaining outcomes. Each value is affine in $l$, so $\mu_G(t,\pi)$ is linear in $l$ and $\sigma^2_G(t,\pi)$ is a quadratic polynomial in $l$; hence $H(t,\pi)=\mu_G(t,\pi)-\eta\sigma^2_G(t,\pi)$ is quadratic in $l$, and collecting the coefficients of $l^2$, $l^1$, $l^0$ gives $c_2$, $c_1$, $c_0$ as stated.

The leading coefficient $c_2$ can be written as $c_2=-\eta\,\mathrm{Var}(Z)$, where $Z$ is the per-unit-quantity payoff equal to $\frac{S(t)-\delta}{2}+r$ with probability $p_A^u+p_B^uq_1^a$ and $\frac{S(t)+\delta}{2}+r$ with probability $p_A^dq_2^a+p_B^dq_1^a$ (and $0$ on the remaining probability mass, which does not affect the variance): $K$ equals $E[Z]$ under these weights, and the bracketed term subtracted from $\eta K^2$ in $c_2$ equals $E[Z^2]$. Since $\eta>0$, this gives $c_2\leq0$, with equality only if $Z$ is degenerate. Thus $H(t,\pi)$ is a concave quadratic in $l$, and whenever $c_2<0$ its unique critical point is a global maximum. Differentiating $c_2l^2+c_1l+c_0$ and setting the result to zero gives $l=-c_1/(2c_2)$; evaluating $H$ at this $l$ gives $c_0-c_1^2/(4c_2)$, establishing both claimed formulas.

All the other cases for Lemma~\ref{lem1} can be computed in a similar way. 
\end{proof}

Figure~\ref{Fig-lem1} illustrates how the optimal placement decision of Corollary~\ref{cor:optimal-l} shifts with the order-flow regime. In the top two panels, the ask side experiences elevated downward pressure or the bid side is simply more active overall; in both cases selling at the best ask remains the dominant choice across the full range of order quantities. This changes in the bottom-left panel, where ask-side orders are both more frequent \emph{and} more likely to execute than bid-side orders: selling at the second-best ask overtakes selling at the best ask, since the higher execution probability at that level outweighs the less favorable price. The bottom-right panel combines a bid-favoring regime with a low risk-aversion parameter ($\eta=0.1$); with the variance penalty nearly negligible, the utility curves are close to linear in $l$, and buying at the second-best bid attains the highest value. Together the four panels show that no single action is uniformly optimal: the ranking depends on the relative order-flow intensities, execution probabilities, and the trader's risk aversion, all of which enter $H(t,\pi)$ through Lemma~\ref{lem1} and Appendix~\ref{appendix:A}.

\begin{figure}[htbp]
\centering
\includegraphics[width=0.48\textwidth]{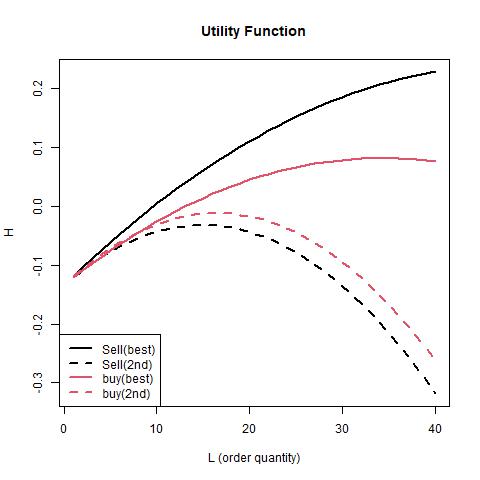}
\includegraphics[width=0.48\textwidth]{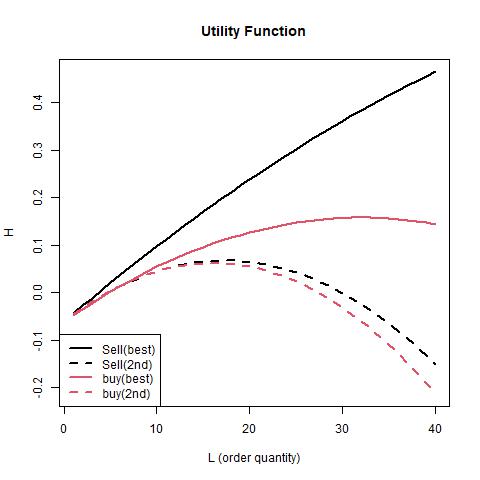}
\includegraphics[width=0.48\textwidth]{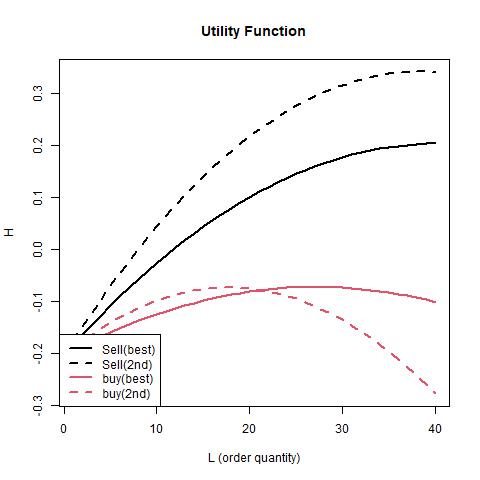}
\includegraphics[width=0.48\textwidth]{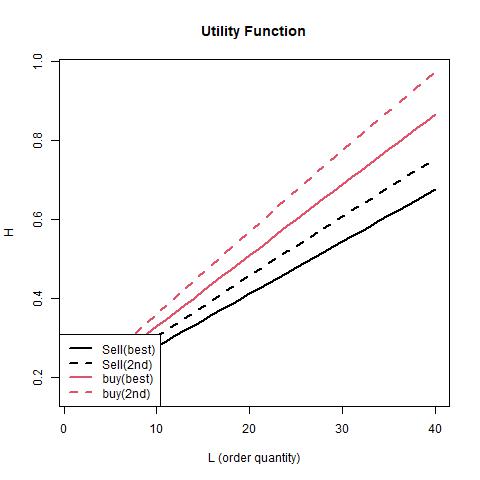}
\caption{The utility function $H(t,\pi)=\mu_G-\eta\sigma_G^2$ as a function of order quantity $l$, for the four resting-order actions of Lemma~\ref{lem1} and Appendix~\ref{appendix:A}: selling at the best or second-best ask, and buying at the best or second-best bid. Each panel uses a different combination of order-flow intensities $\lambda_A^u,\lambda_A^d,\lambda_B^u,\lambda_B^d$, execution probabilities $q_j^i$, and risk-aversion parameter $\eta$.}
\label{Fig-lem1}
\end{figure}

\subsection{Empirical Result}

We illustrate Corollary~\ref{cor:optimal-l}'s closed-form optimal order size using INTC data, in the state that prevails almost all the time: the minimum spread, $S(t)=\delta$, which INTC occupies over $99\%$ of the time (Section~3.3). In this state the spread-narrowing intensities $\lambda_A^d,\lambda_B^u$ are zero by construction, so $p_A^d=p_B^u=0$ and the next price movement is either $A^u$ or $B^d$. Rather than derive $p_A^u$ and $p_B^d$ from the fitted model, whose intensities depend on the recent history of events, we use the observed frequencies: of the 883 events that follow a minimum-spread state in the 10:00--15:30 sample, 402 are $A^u$ and 481 are $B^d$, so $p_A^u\approx0.455$ and $p_B^d\approx0.545$. These are day-averaged frequencies. As a check, we also evaluated $p_A^u(t)$ from the fitted model's own intensities, accumulated from the actual event history, at each of the 883 moments the spread reaches its minimum. The resulting $l^*(t)$ is concentrated close to the static value computed below: $83$ shares at the median, with $86\%$ of evaluations within $5$ shares of that figure, since the fast decay rates ($\beta_1,\beta_2\approx440$--$680$ per second) pull the intensities back toward baseline between the events typically spaced at this state. The remainder, following unusually dense clusters of same-side events, range higher: three moments exceed $200$ shares, and one outlier reaches $7{,}297$. An actual deployment of Corollary~\ref{cor:optimal-l} would capture such moments only by recomputing $l^*$ from live intensities. The mean time until the next price movement from this state is about 22 seconds. The two-tick state, occupied about $1\%$ of the time, is different in kind, since the next movement there is almost surely a narrowing event ($A^d$ in $51\%$ and $B^u$ in $49\%$ of INTC occurrences); we do not illustrate it here. It is also the state in which the cross-side kernel $\Psi$ drives the narrowing intensities (Section~2.2); at the one-tick state illustrated here, $\Psi$ enters only through the widening intensities $\lambda_A^u$ and $\lambda_B^d$, on which the fitted-intensity check above rests.

The execution probabilities $q_1^a,q_2^a$ are not estimated by this model, since doing so would require modeling queue position directly, along the lines of~\citet{figueroalopez2018optimal}. We therefore retain illustrative values, $q_1^a=0.7$ and $q_2^a=0.5$, consistent with Figure~\ref{Fig-lem1} (only $q_1^a$ enters here, since $p_A^d=0$). We take $S(t)=\delta=\$0.01$, $r=\$0.0003$, $\eta=1$, and zero initial inventory ($Y(t)=0$), the last of which sets $c_0=0$ and reduces $c_1$ to $K$ in Corollary~\ref{cor:optimal-l}'s notation.

With these values, $K\approx 0.00406$ and $c_2\approx-2.4\times10^{-5}$, giving
\[
l^*=-\frac{c_1}{2c_2}\approx 85 \text{ shares}, \qquad H(t,\pi)|_{l=l^*}\approx \$0.17.
\]
This is under $1\%$ of the median displayed size at the INTC best ask over the sample (about $11{,}900$ shares), which underlines that the execution probabilities $q_j^i$, illustrative here, are the binding modeling input for realistic order sizes. The small magnitude of $c_2$ means $l^*$ is sensitive to $\eta$: halving $\eta$ doubles $l^*$ (to about 169 shares), since $c_2$ is linear in $\eta$ while $c_1$ does not depend on it.

\section{Conclusion}

In this paper, we have developed a spread-gated Hawkes-flocking model for best bid/ask dynamics, building on the constrained Hawkes framework of {Zheng, Roueff, and Abergel}~\citep{zheng2014modelling} and the flocking mechanism of {Jang, Lee, and Lee}~\citep{jang2020systemic}. We have shown non-explosion via a coupling argument against a dominating Hawkes process (Proposition~\ref{prop:wellposed}), and have given an $O(N)$ recursive maximum likelihood procedure, which we have validated in simulation across three estimation horizons.

Applied to real LOBSTER data for two large-tick stocks, we have found that the cross-side term substantially improves fit: the restriction $\Psi=0$ is rejected on both INTC and MSFT, with AIC and BIC favoring the full model by thousands of points, and goodness-of-fit diagnostics have shown that without $\Psi$ the model reproduces only a small fraction of the observed spread-narrowing events. The full 12-parameter likelihood is multimodal on a single day, so a multi-start search is necessary; at the best fit the observed information is well conditioned in scaled coordinates (Appendix~\ref{appendix:B}).

Several extensions follow naturally. On the empirical side, estimation across more trading days and more stocks, together with out-of-sample comparisons of the full and $\Psi=0$ models, would test the robustness of both findings. On the modeling side, flexible kernels or residual distributions, in the spirit of \citet{lee2025self,lee2026forecasting}, would address the multi-timescale clustering visible in the residual diagnostics. On the application side, direct estimation of the execution probabilities $q_j^i$ via queue-position modeling, along the lines of~\citet{figueroalopez2018optimal}, would let Corollary~\ref{cor:optimal-l}'s illustration use estimated rather than illustrative values, and the single-period placement decision of Section~4 could be extended to a full multi-period trading strategy.

\bibliography{thebib_HawkesFlocking_2026}

\clearpage

\appendix
\section{Appendix: Additional Proof of Lemma~\ref{lem1}}\label{appendix:A}
\begin{align*}
\mu_G(t,\pi(t)=(A^{+}(t),l^a))
=&l^a \left\{\left( \frac{S(t)}{2}+r+\delta \right){(p_A^u q_1^a +p_A^d q_3^a + (p_B^u + p_B^d )q_2^a) }  +\dfrac{\delta}{2}(p_A^dq_3^a + p_B^d q_2^a -p_A^uq_1^a - p_B^u q_2^a)\right\}\\
&+\dfrac{\delta}{2} Y(t) (p_A^u + p_B^u - p_B^d -p_A^d) \\
\sigma^2_G(t,\pi(t)=(A^{+}(t),l^a))
=&\bigg[(\dfrac{S(t)+\delta}{2}+r )l^a+ \frac{\delta}{2}Y(t)\bigg]^2 p_A^u q_1^a  + (\frac{\delta}{2}Y(t))^2 p_A^u (1-q_1^a)\\
   +& \bigg[(\dfrac{S(t)+3\delta}{2}+r)l^a  -  \frac{\delta}{2}Y(t)\bigg]^2 p_A^d q_3^a +  (\frac{\delta}{2}Y(t))^2 p_A^d (1-q_3^a)\\
   +&\bigg[(\dfrac{S(t)+\delta}{2}+r)l^a  + \frac{\delta}{2}Y(t) \bigg]^2 p_B^u q_2^a+ (\frac{\delta}{2}Y(t))^2 p_B^u (1- q_2^a)\\
   +&\bigg[(\dfrac{S(t)+3\delta}{2}+r)l^a -  \frac{\delta}{2}Y(t) \bigg]^2 p_B^d q_2^a+( \frac{\delta}{2}Y(t) )^2   p_B^d (1-q_2^a)-\mu_G(t,\pi(t)=(A^{+}(t),l^a))^2
\end{align*}        

\begin{align*}
\mu_G(t,\pi(t)=(B(t),l^b))
=&l^b \left\{\left( \frac{S(t)}{2}+r \right){(p_A^u q_1^b +p_A^d q_1^b + (p_B^u q_2^b + p_B^d ) )}  +\dfrac{\delta}{2}{(p_A^u q_1^b -p_A^d q_1^b + p_B^uq_2^b - p_B^d ) )} \right\}\\
&+\dfrac{\delta}{2} Y(t) (p_A^u + p_B^u - p_B^d -p_A^d) \\
\sigma^2_G(t,\pi(t)=(B(t),l^b))
&=\bigg[(\dfrac{S(t)+\delta}{2}+r)l^b+ \frac{\delta}{2}Y(t)\bigg]^2(p_A^u q_1^b +( \frac{\delta}{2}Y(t))^2 p_A^u (1-q_1^b)\\
   +& \bigg[(\dfrac{S(t)-\delta}{2}+r)l^b-  \frac{\delta}{2}Y(t) \bigg]^2 p_A^d q_1^b+( \frac{\delta}{2}Y(t))^2p_A^d(1- q_1^b)\\
   +&\bigg[(\dfrac{S(t)+\delta}{2}+r)l^b + \frac{\delta}{2}Y(t) \bigg]^2 p_B^u q_2^b+( \frac{\delta}{2}Y(t))^2 p_B^u(1- q_2^b)\\
   +&\bigg[(\dfrac{S(t)-\delta}{2}+r)l^b -  \frac{\delta}{2}Y(t)\bigg]^2 p_B^d-\mu_G(t,\pi(t)=(B(t),l^b))^2
\end{align*}     

\begin{align*}
\mu_G(t,\pi(t)=(B^{-}(t),l^b))
=&l^b \left\{\left( \frac{S(t)}{2}+r +\delta\right){(p_A^u q_2^b +p_A^d q_2^b + (p_B^u q_3^b + p_B^d q_1^b ) )}  +\dfrac{\delta}{2}{(p_A^u q_2^b -p_A^d q_2^b + p_B^uq_3^b - p_B^dq_1^b ) )} \right\}\\
&+\dfrac{\delta}{2} Y(t) (p_A^u + p_B^u - p_B^d -p_A^d) \\
\sigma^2_G(t,\pi(t)=(B^{-}(t),l^b))
&=\bigg[(\dfrac{S(t)+3\delta}{2}+r)l^b+ \frac{\delta}{2}Y(t)\bigg]^2(p_A^u q_2^b +( \frac{\delta}{2}Y(t))^2 p_A^u (1-q_2^b)\\
   +& \bigg[(\dfrac{S(t)+\delta}{2}+r)l^b-  \frac{\delta}{2}Y(t) \bigg]^2 p_A^d q_2^b+( \frac{\delta}{2}Y(t))^2p_A^d(1- q_2^b)\\
   +&\bigg[(\dfrac{S(t)+3\delta}{2}+r)l^b + \frac{\delta}{2}Y(t) \bigg]^2 p_B^u q_3^b+( \frac{\delta}{2}Y(t))^2 p_B^u(1- q_3^b)\\
   +&\bigg[(\dfrac{S(t)+\delta}{2}+r)l^b -  \frac{\delta}{2}Y(t)\bigg]^2 p_B^d q_1^b +( \frac{\delta}{2}Y(t))^2 p_B^d(1- q_1^b)\\
   &-\mu_G(t,\pi(t)=(B^{-}(t),l^b))^2
\end{align*}     

\begin{proof}

In the main proof, we showed the case in which the trader places a limit sell order at the price $A(t)$, with quantity $l^a$. Let us consider the rest of the cases. 

\noindent \textbf{Order placement (action): $(A^{+}(t), l^a)$}

\begin{itemize}
\item If next event is $A^u$, there is a chance of execution, or the trader's order has been moved to the first best ask price. The chance of execution is $q_1^a$. 
\item If next event is $A^d$, there is a chance of execution, or the trader's order has been moved to the third best ask price. Let us denote the chance of execution as $q_3^a$.
\item If next event is $B^u$ or $B^d$, there is a chance of execution, or the trader's order is still at the second best ask. Let us denote the probability of execution of the trader's order in this case as $q_2^a$.
\end{itemize}

Then, after the first price movement ($t+\tau$), 

\begin{table}[htbp]
\begin{tabular}{c|c|c|c|c}
Event & $X(t+\tau)-X(t)$ & $Y(t+\tau)$ & $A(t+\tau)+B(t+\tau)$ & $G(t+\tau)-G(t)$\\
\hline
$A^u$,  $E$  & $(A(t)+r+\delta) l^a$  & $Y(t)-l^a$  & $A(t)+B(t)+\delta$  & $(\dfrac{S(t)+\delta}{2}+r)l^a + \frac{\delta}{2}Y(t)$\\
$A^u$,  $E^c$  &  0 & $Y(t)$  & $A(t)+B(t)+\delta$  & $\frac{\delta}{2}Y(t)$\\
\hline
$A^d$, $E$ & $(A(t)+r+\delta) l^a$ & $Y(t)-l^a$ & $A(t)+B(t)-\delta$ & $(\dfrac{S(t)+3\delta}{2}+r)l^a -  \frac{\delta}{2}Y(t)$\\
$A^d$, $E^c$ & 0  & $Y(t)$  & $A(t)+B(t)-\delta$  & $- \frac{\delta}{2}Y(t)$\\
\hline
$B^u$, $E$ & $(A(t)+r+\delta) l^a$ & $Y(t)-l^a$ & $A(t)+B(t)+\delta$ &$(\dfrac{S(t)+\delta}{2}+r)l^a  + \frac{\delta}{2}Y(t)$\\
$B^u$, $E^c$ & 0  & $Y(t)$ & $A(t)+B(t)+\delta$ & $+ \frac{\delta}{2}Y(t)$\\
\hline
$B^d$, $E$ & $(A(t)+r+\delta) l^a$ & $Y(t)-l^a$ & $A(t)+B(t)-\delta$ & $(\dfrac{S(t)+3\delta}{2}+r)l^a -  \frac{\delta}{2}Y(t)$\\
$B^d$, $E^c$ & 0  & $Y(t)$ & $A(t)+B(t)-\delta$ & $-  \frac{\delta}{2}Y(t)$
\end{tabular}
\caption{Cash flow and inventory movements after the action $(A^{+}(t), l^a)$, for each events of price movements up/down, and event E (order executed)}\label{pf:tb2}
\end{table}

Using the results from Table \ref{pf:tb2}, we have

\begin{align*}
\mu_G(t,\pi(t)=(A^{+}(t),l^a))
=&l^a \left\{\left( \frac{S(t)}{2}+r+\delta \right){(p_A^u q_1^a +p_A^d q_3^a + (p_B^u + p_B^d )q_2^a) }  +\dfrac{\delta}{2}(p_A^dq_3^a + p_B^d q_2^a -p_A^uq_1^a - p_B^u q_2^a)\right\}\\
&+\dfrac{\delta}{2} Y(t) (p_A^u + p_B^u - p_B^d -p_A^d) \\
\sigma^2_G(t,\pi(t)=(A^{+}(t),l^a))
=&\bigg[(\dfrac{S(t)+\delta}{2}+r )l^a+ \frac{\delta}{2}Y(t)\bigg]^2 p_A^u q_1^a  + (\frac{\delta}{2}Y(t))^2 p_A^u (1-q_1^a)\\
   +& \bigg[(\dfrac{S(t)+3\delta}{2}+r)l^a  -  \frac{\delta}{2}Y(t)\bigg]^2 p_A^d q_3^a +  (\frac{\delta}{2}Y(t))^2 p_A^d (1-q_3^a)\\
   +&\bigg[(\dfrac{S(t)+\delta}{2}+r)l^a  + \frac{\delta}{2}Y(t) \bigg]^2 p_B^u q_2^a+ (\frac{\delta}{2}Y(t))^2 p_B^u (1- q_2^a)\\
   +&\bigg[(\dfrac{S(t)+3\delta}{2}+r)l^a -  \frac{\delta}{2}Y(t) \bigg]^2 p_B^d q_2^a+( \frac{\delta}{2}Y(t) )^2   p_B^d (1-q_2^a)-\mu_G(t,\pi(t)=(A^{+}(t),l^a))^2
\end{align*}        

\noindent \textbf{Order placement (action): $(B(t), l^b)$}

\begin{itemize}
\item If next event is $B^d$, this event implies that the trader's order placed at $B(t)$ has been executed with probability 1. 
\item If next event is $B^u$, there is a chance of execution, or the trader's order has been moved to the second best bid. The chance of execution is $q_2^b$. 
\item If next event is $A^u$ or $A^d$, there is a chance of execution. Let us denote the probability of execution of the trader's order in this case as $q_1^b$.
\end{itemize}

Then, after the first price movement ($t+\tau$), 

\begin{table}[htbp]
\begin{tabular}{c|c|c|c|c}
Event & $X(t+\tau)-X(t)$ & $Y(t+\tau)$ & $A(t+\tau)+B(t+\tau)$ & $G(t+\tau)-G(t)$\\
\hline
$A^u$,  $E$  & $(-B(t)+r) l^b$  & $Y(t)+l^b$  & $A(t)+B(t)+\delta$  & $(\dfrac{S(t)+\delta}{2}+r)l^b + \frac{\delta}{2}Y(t)$\\
$A^u$,  $E^c$  &  0 & $Y(t)$  & $A(t)+B(t)+\delta$  & $\frac{\delta}{2}Y(t)$\\
\hline
$A^d$, $E$ & $(-B(t)+r) l^b$ & $Y(t)+l^b$ & $A(t)+B(t)-\delta$ & $(\dfrac{S(t)-\delta}{2}+r)l^b -  \frac{\delta}{2}Y(t)$\\
$A^d$, $E^c$ & 0  & $Y(t)$  & $A(t)+B(t)-\delta$  & $- \frac{\delta}{2}Y(t)$\\
\hline
$B^u$, $E$ & $(-B(t)+r) l^b$ & $Y(t)+l^b$ & $A(t)+B(t)+\delta$ &$(\dfrac{S(t)+\delta}{2}+r)l^b  + \frac{\delta}{2}Y(t)$\\
$B^u$, $E^c$ & 0  & $Y(t)$ & $A(t)+B(t)+\delta$ & $+ \frac{\delta}{2}Y(t)$\\
\hline
$B^d$ & $(-B(t)+r) l^b$ & $Y(t)+l^b$ & $A(t)+B(t)-\delta$ & $(\dfrac{S(t)-\delta}{2}+r)l^b -  \frac{\delta}{2}Y(t)$
\end{tabular}
\caption{Cash flow and inventory movements after the action $(B(t), l^b)$, for each events of price movements up/down, and event E (order executed)}\label{pf:tb3}
\end{table}

Using the results from Table \ref{pf:tb3}, we have

\begin{align*}
\mu_G(t,\pi(t)=(B(t),l^b))
=&l^b \left\{\left( \frac{S(t)}{2}+r \right){(p_A^u q_1^b +p_A^d q_1^b + (p_B^u q_2^b + p_B^d ) )}  +\dfrac{\delta}{2}{(p_A^u q_1^b -p_A^d q_1^b + p_B^uq_2^b - p_B^d ) )} \right\}\\
&+\dfrac{\delta}{2} Y(t) (p_A^u + p_B^u - p_B^d -p_A^d) \\
\sigma^2_G(t,\pi(t)=(B(t),l^b))
&=\bigg[(\dfrac{S(t)+\delta}{2}+r)l^b+ \frac{\delta}{2}Y(t)\bigg]^2(p_A^u q_1^b +( \frac{\delta}{2}Y(t))^2 p_A^u (1-q_1^b)\\
   +& \bigg[(\dfrac{S(t)-\delta}{2}+r)l^b-  \frac{\delta}{2}Y(t) \bigg]^2 p_A^d q_1^b+( \frac{\delta}{2}Y(t))^2p_A^d(1- q_1^b)\\
   +&\bigg[(\dfrac{S(t)+\delta}{2}+r)l^b + \frac{\delta}{2}Y(t) \bigg]^2 p_B^u q_2^b+( \frac{\delta}{2}Y(t))^2 p_B^u(1- q_2^b)\\
   +&\bigg[(\dfrac{S(t)-\delta}{2}+r)l^b -  \frac{\delta}{2}Y(t)\bigg]^2 p_B^d-\mu_G(t,\pi(t)=(B(t),l^b))^2
\end{align*}     

\noindent \textbf{Order placement (action): $(B^{-}(t), l^b)$}

\begin{itemize}
\item If next event is $B^d$, there is a chance of execution, or the trader's order has been moved to the first best bid price. The chance of execution is $q_1^b$. 
\item If next event is $B^u$, there is a chance of execution, or the trader's order has been moved to the third best bid price. Let us denote the chance of execution as $q_3^b$.
\item If next event is $A^u$ or $A^d$, there is a chance of execution, or the trader's order is still at the second best bid. Let us denote the probability of execution of the trader's order in this case as $q_2^b$.
\end{itemize}

Then, after the first price movement ($t+\tau$), 

\begin{table}[htbp]
\begin{tabular}{c|c|c|c|c}
Event & $X(t+\tau)-X(t)$ & $Y(t+\tau)$ & $A(t+\tau)+B(t+\tau)$ & $G(t+\tau)-G(t)$\\
\hline
$A^u$,  $E$  & $(-B(t)+r+\delta) l^b$  & $Y(t)+l^b$  & $A(t)+B(t)+\delta$  & $(\dfrac{S(t)+3\delta}{2}+r)l^b + \frac{\delta}{2}Y(t)$\\
$A^u$,  $E^c$  &  0 & $Y(t)$  & $A(t)+B(t)+\delta$  & $\frac{\delta}{2}Y(t)$\\
\hline
$A^d$, $E$ & $(-B(t)+r+\delta) l^b$ & $Y(t)+l^b$ & $A(t)+B(t)-\delta$ & $(\dfrac{S(t)+\delta}{2}+r)l^b -  \frac{\delta}{2}Y(t)$\\
$A^d$, $E^c$ & 0  & $Y(t)$  & $A(t)+B(t)-\delta$  & $- \frac{\delta}{2}Y(t)$\\
\hline
$B^u$, $E$ & $(-B(t)+r+\delta) l^b$ & $Y(t)+l^b$ & $A(t)+B(t)+\delta$ &$(\dfrac{S(t)+3\delta}{2}+r)l^b  + \frac{\delta}{2}Y(t)$\\
$B^u$, $E^c$ & 0  & $Y(t)$ & $A(t)+B(t)+\delta$ & $+ \frac{\delta}{2}Y(t)$\\
\hline
$B^d$, $E$ & $(-B(t)+r+\delta) l^b$ & $Y(t)+l^b$ & $A(t)+B(t)-\delta$ & $(\dfrac{S(t)+\delta}{2}+r)l^b -  \frac{\delta}{2}Y(t)$ \\
$B^d$, $E^c$ & 0 & $Y(t)$ & $A(t)+B(t)-\delta$ & $- \frac{\delta}{2}Y(t)$
\end{tabular}
\caption{Cash flow and inventory movements after the action $(B^{-}(t), l^b)$, for each events of price movements up/down, and event E (order executed)}\label{pf:tb4}
\end{table}

Using the results from Table \ref{pf:tb4}, we have

\begin{align*}
\mu_G(t,\pi(t)=(B^{-}(t),l^b))
=&l^b \left\{\left( \frac{S(t)}{2}+r +\delta\right){(p_A^u q_2^b +p_A^d q_2^b + (p_B^u q_3^b + p_B^d q_1^b ) )}  +\dfrac{\delta}{2}{(p_A^u q_2^b -p_A^d q_2^b + p_B^uq_3^b - p_B^dq_1^b ) )} \right\}\\
&+\dfrac{\delta}{2} Y(t) (p_A^u + p_B^u - p_B^d -p_A^d) \\
\sigma^2_G(t,\pi(t)=(B^{-}(t),l^b))
&=\bigg[(\dfrac{S(t)+3\delta}{2}+r)l^b+ \frac{\delta}{2}Y(t)\bigg]^2(p_A^u q_2^b +( \frac{\delta}{2}Y(t))^2 p_A^u (1-q_2^b)\\
   +& \bigg[(\dfrac{S(t)+\delta}{2}+r)l^b-  \frac{\delta}{2}Y(t) \bigg]^2 p_A^d q_2^b+( \frac{\delta}{2}Y(t))^2p_A^d(1- q_2^b)\\
   +&\bigg[(\dfrac{S(t)+3\delta}{2}+r)l^b + \frac{\delta}{2}Y(t) \bigg]^2 p_B^u q_3^b+( \frac{\delta}{2}Y(t))^2 p_B^u(1- q_3^b)\\
   +&\bigg[(\dfrac{S(t)+\delta}{2}+r)l^b -  \frac{\delta}{2}Y(t)\bigg]^2 p_B^d q_1^b +( \frac{\delta}{2}Y(t))^2 p_B^d(1- q_1^b)\\
   &-\mu_G(t,\pi(t)=(B^{-}(t),l^b))^2
\end{align*}     

\end{proof}

\newpage
\section{Appendix: Additional Real-Data Estimation Results}\label{appendix:B}

This appendix documents the estimation diagnostics referenced in Section~3.3: the multi-start search, the conditioning of the observed information at the best fit, three restricted specifications, a stronger $\Psi=0$ test with separate baselines, and MSFT's goodness-of-fit figure.

\subsection*{Multi-start diagnostics}

For each stock, $\theta$ was re-estimated from many independent starting points (log-uniform random initializations, drawn independently per parameter), since a single optimization run was found to be unreliable on this likelihood surface. Table~\ref{table:appendixB-spread} reports the resulting spread in log-likelihood at convergence, and the condition number of the observed information matrix at the best fit found for each stock.

\begin{table}[htbp]
\centering
\begin{tabular}{l | c c}
 & INTC & MSFT \\
\hline
Starting points used & 30 & 24 \\
Best log-likelihood found & $-572.5$ & $107.5$ \\
Log-likelihood range across starts & $769.8$ & $1472.2$ \\
Information matrix condition number (scaled), best fit & $486$ & $449$ \\
\end{tabular}
\caption{Multi-start estimation diagnostics, full 12-parameter model. A log-likelihood range this wide among nominally converged fits indicates a multimodal likelihood surface, so a single optimization run is unreliable. The condition number is that of the observed information matrix in coordinates scaled by each parameter family's typical size ($0.02$ for $\mu_i$, $500$ for $\beta_i$, $100$ for the $\alpha$'s), evaluated with coefficients below $0.5$ moved to $0.5$ so that finite differences stay inside the parameter domain.}
\label{table:appendixB-spread}
\end{table}

To see which directions of the parameter space are least well determined, we computed the eigendecomposition of the observed information matrix in scaled coordinates at each stock's best fit. All eigenvalues are positive, ranging from $5.7$ to $2.8\times10^{3}$ for INTC and from $5.7$ to $2.6\times10^{3}$ for MSFT. In both stocks the least well determined direction is dominated by the bid-side flocking coefficient $\alpha_{2w}$, whose one-day standard error is about $40$ units, or about $8\%$ of the estimate for INTC and $7\%$ for MSFT. With independent days pooled, standard errors shrink with the square root of the number of days, which is the basis for the sample-size figure in Section~3.3.

We also estimated three restricted specifications, each with the same multi-start protocol as the full model, and compared them with the unconstrained fit by a likelihood-ratio test (Table~\ref{table:appendixB-reductions}).

\begin{table}[htbp]
\centering
\small
\begin{tabular}{llc}
\toprule
Constraint tested & Stock & LR test \\
\midrule
Self-excitation $=0$ ($\alpha_{1s},\alpha_{2s}$) & INTC & inconclusive$^*$ \\
Narrow-flocking $=0$ ($\alpha_{1n},\alpha_{2n}$) & MSFT & Rejected, $p<0.0001$ \\
Aliasing ($\alpha_{1c}{=}\alpha_{1w}$, $\alpha_{2c}{=}\alpha_{2w}$) & INTC & Rejected, $p<0.0001$ \\
Aliasing ($\alpha_{1c}{=}\alpha_{1w}$, $\alpha_{2c}{=}\alpha_{2w}$) & MSFT & Rejected, $p<0.0001$ \\
\bottomrule
\end{tabular}
\caption{Three restricted specifications (10 parameters), each re-estimated with the same multi-start protocol as the full model. $^*$Negative LR statistic: the full model's own multi-start search had not found its best fit on this multimodal surface, so no formal comparison was possible.}
\label{table:appendixB-reductions}
\end{table}

The multimodality is a property of the optimization rather than of the local curvature at the best fit, which is why every estimate reported in Section~3.3 comes from a multi-start search using both gradient-based and derivative-free optimization. The point estimates in Table~\ref{table:realdata-estimates}, the likelihood-ratio test for $\Psi=0$ (Table~\ref{table:realdata-tests}), and the goodness-of-fit comparison below all compare fitted likelihoods and do not rely on the individual standard errors.

\subsection*{Separate-baseline $\Psi=0$ test}

The test in Section~3.3 (Table~\ref{table:realdata-tests}) shares one baseline intensity between the widening and narrowing event types on each side of the book ($\mu_1$ for $A^u,A^d$; $\mu_2$ for $B^u,B^d$), so it does not on its own separate the effect of $\Psi$ from that of any other mechanism that could make widening and narrowing rates differ. As a stronger test, we re-estimated both the full model and the $\Psi=0$ null with four separate baselines, one per event type, using the same multi-start protocol (30 starts per stock per model). Table~\ref{table:appendixB-strongernull} reports the result.

\begin{table}[htbp]
\centering
\begin{tabular}{l | c c | c c}
 & \multicolumn{2}{c|}{INTC} & \multicolumn{2}{c}{MSFT} \\
 & 14 param & $\Psi=0$ (10 param) & 14 param & $\Psi=0$ (10 param) \\
\hline
Log-likelihood & $-458.4$ & $-3035.6$ & $256.5$ & $-3286.5$ \\
AIC & $944.9$ & $6091.2$ & $-485.0$ & $6592.9$ \\
BIC & $1021.6$ & $6146.0$ & $-402.9$ & $6651.6$ \\
\hline
LR statistic (df=4) & \multicolumn{2}{c|}{$5{,}154.3$} & \multicolumn{2}{c}{$7{,}085.9$} \\
\end{tabular}
\caption{Stronger $\Psi=0$ test, both stocks: four separate baselines (one per event type) instead of one per side. As in Table~\ref{table:realdata-tests}, the flocking parameters sit at a boundary under $H_0$, so the $\chi^2_4$ reference is not exact. The LR statistics are smaller than in Table~\ref{table:realdata-tests} because some of the widening/narrowing asymmetry is now absorbed by the baselines. The restriction $\Psi=0$ is nonetheless still rejected, and the full model still improves fit substantially, at both stocks.}
\label{table:appendixB-strongernull}
\end{table}

\subsection*{MSFT goodness-of-fit}

Figure~\ref{fig:appendixB-msft-qq} reports the time-rescaled residual QQ-plots for MSFT, full model versus $\Psi=0$, complementing the INTC figure (Figure~\ref{fig:qqplot-gof}) in the main text. The pattern is the same as for INTC: $A^d$ and $B^u$ show markedly worse fit under $\Psi=0$, while $A^u$ and $B^d$ are comparatively stable across both models.

\begin{figure}[htbp]
\centering
\includegraphics[width=0.95\textwidth]{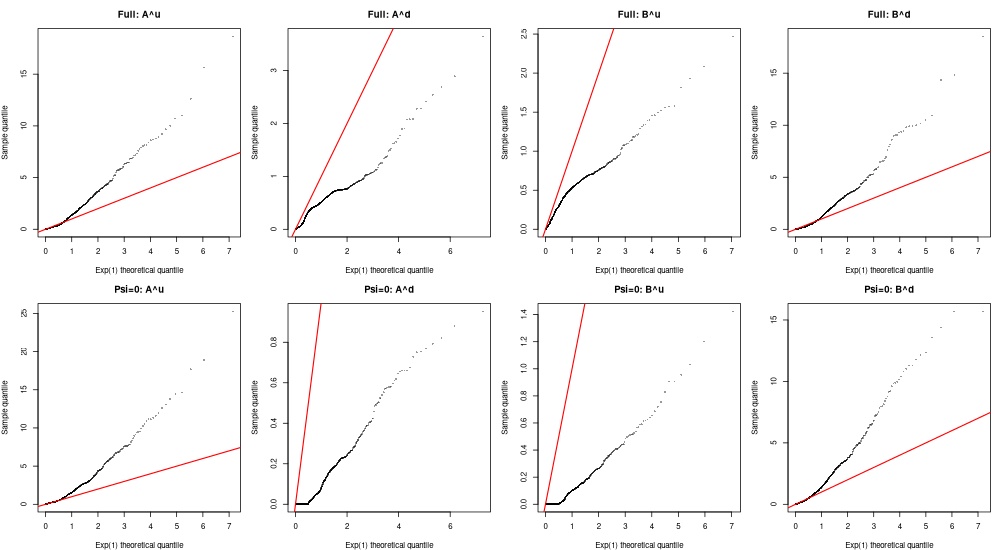}
\caption{Time-rescaled residual QQ-plots against Exponential(1), MSFT: full model (top row) versus $\Psi=0$ (bottom row), one panel per event type. The red line is the identity.}
\label{fig:appendixB-msft-qq}
\end{figure}

\end{document}